\documentclass{elsarticle}

\usepackage{amsthm}
\usepackage{graphicx}
\usepackage[utf8]{inputenc}
\usepackage{amsmath}
\usepackage{amssymb}
\usepackage{latexsym}
\usepackage{enumerate}
\usepackage[normalem]{ulem}
\usepackage{color}
\usepackage{float}
\usepackage{subfigure}
\usepackage[pagewise]{lineno}

\newcommand{\edge}[2]{#1 #2}

\newcommand{\NP}{\ensuremath{\mathcal{NP}}}

\def\O(#1){\ensuremath{\mathcal{O}(#1)}}

\journal{Theoretical Computer Science}

\newtheorem{theorem}{Theorem}
\newtheorem{lemma}{Lemma}
\newtheorem{corollary}{Corollary}
\newtheorem{definition}{Definition}

\def\squareforqed{\hbox{\rlap{$\sqcap$}$\sqcup$}}
\def\qed{\ifmmode\squareforqed\else{\unskip\nobreak\hfil
\penalty50\hskip1em\null\nobreak\hfil\squareforqed
\parfillskip=0pt\finalhyphendemerits=0\endgraf}\fi}

\makeatletter
\newif\if@restonecol
\makeatother

\begin{document}

\begin{frontmatter}

\title{Fan-Crossing Free Graphs   and \\ Their Relationship to other Beyond-Planar Graphs\tnoteref{dfg}}
\tnotetext[dfg]{Supported by the Deutsche Forschungsgemeinschaft
(DFG),
  grant Br835/20-1}

\author{Franz J.\ Brandenburg}
\ead{brandenb@informatik.uni-passau.de}
\address{94030 Passau, Germany}

\begin{abstract}
A graph is \emph{fan-crossing free} if it has a drawing in the plane
so that each edge is crossed by independent edges, that is the
crossing edges have distinct vertices. On the other hand, it is
\emph{fan-crossing} if the crossing edges have a common vertex, that
is they form a fan. Both are prominent  examples for   beyond-planar
graphs. Further well-known beyond-planar classes  are the
$k$-planar, $k$-gap-planar, quasi-planar, and right angle crossing
graphs.

We use  the subdivision,  node-to-circle expansion and path-addition
operations to distinguish all these graph classes. In particular, we
show that the 2-subdivision and the node-to-circle expansion of any
graph is fan-crossing free, which does not hold for fan-crossing and
$k$-(gap)-planar graphs, respectively.
 Thereby, we obtain graphs that are
 fan-crossing free and neither fan-crossing nor $k$-(gap)-planar.

 Finally, we show that some graphs have a  unique
 fan-crossing free embedding, that there are thinned  maximal fan-crossing
free graphs, and  that the recognition problem for fan-crossing free
graphs is NP-complete.

\end{abstract}


\end{frontmatter}

\section{Introduction}

 We consider graphs $G$ 
 that are \emph{simple} both in a graph theoretic
and in a topological sense. Thus there are no multi-edges or loops,
adjacent edges do not cross, and two edges cross at most once in a
drawing. Graphs are often defined by particular properties of a
drawing. The planar graphs, in which  edge crossings are excluded,
are the most prominent example. There has been  recent interest in
the study of   \emph{beyond-planar graphs}
\cite{dlm-survey-beyond-19, ht-beyond-book-20,klm-bib-17}, which are
generally defined by drawings with specific restrictions on
crossings.
%

A graph is $k$-\emph{planar} if it has a drawing in the plane so
that each edge is   crossed by at most $k$ edges. It is
\emph{fan-crossing free} if the crossing edges are independent,
i.e., they  have distinct vertices,
 and \emph{fan-crossing} if the crossing edges have a common
vertex, i.e., they form a fan.
%
%
  A drawing of a graph is
$k$-\emph{quasi-planar} if $k$ edges do not mutually cross.
3-quasi-planar graphs are called quasi-planar, see
Fig.~\ref{fig:grids}. The aforementioned graphs can also be defined
by first order logic formulas \cite{b-FOL-18} and in terms of an
avoidance of (natural and radial) grids \cite{afps-grids-14}. A
drawing is $k$-\emph{gap-planar} \cite{bbc-1gap-18} if each crossing
is assigned to an edge involved in the crossing so that  at most $k$
crossings are assigned to each edge. Then every $2k$-planar drawing
is $k$-gap-planar \cite{bbc-1gap-18}. These properties are
topological. They hold for embeddings, which are equivalence classes
of topologically equivalent drawings.
 Right angle crossing  (RAC) is  a geometric property, in which the edges are drawn straight line
and may cross at a right angle \cite{del-dgrac-11}.
%

\begin{figure}[t]
  \centering
  \subfigure[]{
    \includegraphics[scale=0.25]{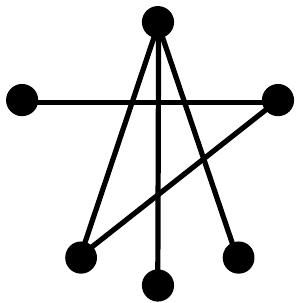}
    }
    \hspace{5mm}
    \subfigure[]{
        \includegraphics[scale=0.25]{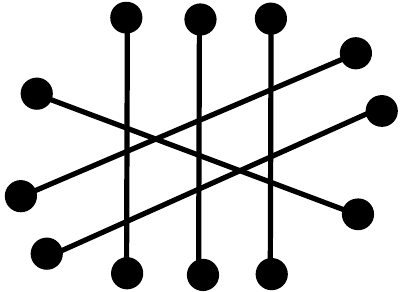}
} \hspace{5mm}
 \subfigure[]{
        \includegraphics[scale=0.25]{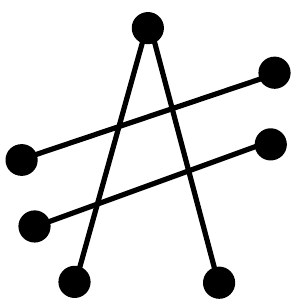}
        }
        \hspace{5mm}
        \subfigure[]{
        \includegraphics[scale=0.25]{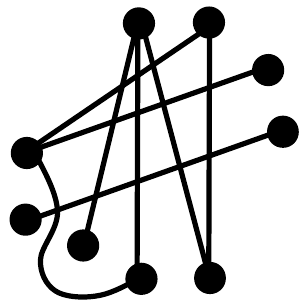}
  }
\caption{(a) A fan crossing, (b) crossings of independent edges (c)
a 2-planar crossing, and (d) a quasi-planar crossing
  }
  \label{fig:grids}
\end{figure}

 The classes of $k$-planar, fan-crossing free,
fan-crossing, quasi-planar, $k$-gap-planar, and right-angle crossing
graphs  
are denoted by $k$\textsc{PLANAR},
\textsc{FCF}, \textsc{FAN},
\textsc{QUASI},  $k$\textsc{GAP}, and \textsc{RAC},  
 respectively. Then 0\textsc{PLANAR}$=$ 0\textsc{GAP} is the class of
 planar graphs.
These and other classes have been studied with different intensity
and depth. In particular, the \emph{density}
\cite{dlm-survey-beyond-19, ht-beyond-book-20}, which is  the
maximal number of edges of $n$-vertex graphs, and the size of the
largest complete (bipartite) graph  \cite{abks-beyond-Kn-19} are
well explored, see Table~\ref{table1} and Fig.~\ref{fig:inclusion}.
 It has been shown that the recognition problem is
NP-complete for 1-planar \cite{GB-AGEFCE-07, km-mo1ih-13},
fan-crossing (fan-planar) \cite{bddmpst-fan-15}, 1-gap-planar
\cite{bbc-1gap-18}, and RAC graphs \cite{abs-RACNP-12}, whereas a
proof for the NP-completeness is still missing for $k$-planar,
$k$-gap-planar, fan-crossing free, and quasi-planar graphs if $k
\geq 2$ \cite{dlm-survey-beyond-19}.

Consider relationships among the aforementioned graph classes.
   Fan-crossing
and fan-crossing free are complementary on edges that are crossed at
least twice   \cite{b-FOL-18}. Both admit edges that are uncrossed
or are crossed once. Hence, every 1-planar graph is both
fan-crossing and fan-crossing free, but not conversely, since there
are graphs that are fan-crossing and fan-crossing free, but   not
1-planar \cite{b-fan+fcf-18}. There are fan-crossing graphs that are
not fan-crossing free, since the density of fan-crossing free graphs
is $4n-8$ \cite{cpkk-fan-15}, whereas there are $n$-vertex
fan-crossing graphs with $5n-10$ edges \cite{b-fan-20,ku-dfang-14}.
  Particular examples for fan-crossing and non-fan-crossing
free graphs are $K_7$ and $K_{p,q}$ with $p=3,4$ and $q \geq 7$
\cite{abks-beyond-Kn-19}. We prove the opposite direction, which was
open so far \cite{dlm-survey-beyond-19}. In particular, we show that
the 2-subdivision of $K_n$ for $n \geq 12$ and $d$-dimensional
cube-connected cycles $CCC^{d}$ for $d \geq 11$ are fan-crossing
free   but neither fan-crossing  nor 1-gap-planar. In addition, for
every $k \geq 0$, there are  graphs that are RAC,  and thus
fan-crossing free, but not $k$-gap-planar, and thus not $2k$-planar.

 Our tools,  are the  subdivision, node-to-circle expansion  and
path-addition operations. An $s$-\emph{subdivision} expands an edge
into a path of length   $s+1$. In other terms the edge has  $s$
bends.
In an $s$-subdivision of graph $G$, each edge of graph $G$ has a $j$-subdivision for $j \leq s$.
For subdivision of  $G$, let  $s$ be the number of edges.
 A \emph{node-to-circle expansion}
replaces each vertex $v$ of degree $d$ by a circle of length $d$, so
that each edge incident to $v$ is inherited by a vertex of the
circle. A \emph{path-addition} adds a (long) path between any two
vertices of  graph $G$. These operations distinguish the above graph
classes, so that some are closed and others are not. Even more, some
classes are \emph{universal} for $s$-subdivision ($s \leq 3$) and
node-to-circle expansion, respectively, that is the image of any
graph  is in the particular class.

  Many more classes of beyond-planar graphs have been studied, see
\cite{dlm-survey-beyond-19,ht-beyond-book-20,klm-bib-17}, some of
which are relevant for this paper. A 1-planar drawing is
IC-\emph{planar} \cite{a-cnircn-08, bbhnr-NIC-17, bdeklm-IC-16} if
each vertex is incident to  at most one crossed edge. It is
\emph{fan-planar} \cite{bcghk-rfpg-17, bddmpst-fan-15, ku-dfang-14}
if it is fan-crossing and excludes crossings of an edge from both
sides, called \emph{configuration II} in \cite{ku-dfang-14}, see
Fig.~\ref{fig:fanplanar}(c). Configuration II is a
 restriction, since there are fan-crossing graphs that are
not fan-planar \cite{b-fan-20}. However, it has no impact on the
density \cite{b-fan-20} and the results on fan-planar graphs proved
in \cite{abks-beyond-Kn-19, bcghk-rfpg-17,bddmpst-fan-15}.\\

Fan-crossing free graphs were introduced  by Cheong et
al.~\cite{cpkk-fan-15}. They focus on the  density of
($k$-)fan-crossing free graphs. Complete and complete bipartite
graphs were studied by Angelini et al.~\cite{abks-beyond-Kn-19}. The
state of the art is as follows, see also
\cite{dlm-survey-beyond-19}.

\begin{enumerate}
 \item Every $n$-vertex fan-crossing free graph has at most $4n-8$
 edges  and the bound is tight.
\item Every fan-crossing free drawing of a graph with $4n-8$ edges is 1-planar.
\item Every fan-crossing free graph with a drawing with straight line edges has at
most $4n-9$ edges and the bound is tight.
 \item 1-planar graphs and right angle crossing (RAC) graphs are
 fan-crossing free.
 \item $K_6$ is fan-crossing free, whereas $K_7$ is not.
 $K_{p,q}$ is fan-crossing free if and only if $p \leq 2$ or $p \leq
 4$ and $q \leq 6$.
\end{enumerate}

For fan-planar graphs, the following has been proved
\cite{abks-beyond-Kn-19, bcghk-rfpg-17, bddmpst-fan-15, b-fan-20,
ku-dfang-14}.  The shown results also hold for fan-crossing graphs,
since the restriction from configuration II is not used in the
proofs or it does not matter.

\begin{enumerate}
 \item Every $n$-vertex fan-crossing  graph has at most $5n-10$
 edges  and the bound is tight.
 \item Every fan-crossing graph is quasi-planar.
\item The fan-crossing graphs and the 2-planar and 1-gap-planar graphs, respectively,  are
incomparable.
\item The recognition problem for fan-crossing graphs is NP-complete.
\item $K_7$ is fan-crossing, whereas $K_8$ is not.
 $K_{p,q}$ is fan-crossing if and only if $\min\{p,q\} \leq 4$.
\end{enumerate}

Similar facts are known for 1-planar   \cite{klm-bib-17}, RAC
  \cite{del-dgrac-11, el-racg1p-13}, 1-gap-planar
\cite{bbc-1gap-18}, and quasi-planar graphs \cite{at-mneqpg-07,
aapps-qpg-97, abbddd-quasiplanar-20}, respectively, as surveyed in
\cite{dlm-survey-beyond-19}.\\

 In this work, we add some more facts on   fan-crossing free graphs,
 which  demonstrate the power of these graphs. Some of the results come
 as expected, for example the NP-hardness, which is  stated as an open problem  in
\cite{dlm-survey-beyond-19}. As our main contribution, we  study
three graph operations, which are used to distinguish fan-crossing
free graphs from other well-known classes of beyond-planar graphs
including fan-crossing graphs.
 In particular, we prove the following:

\begin{enumerate}
\item
The 2-subdivision and the node-to-circle expansion,
respectively, of any graph is fan-crossing free, so that the
fan-crossing free graphs are universal for these operations. The
quasi-planar graphs  are universal for 1-subdivision and
node-to-circle expansion.
The $k$-planar, $k$-gap-planar and   fan-crossing graphs are
 universal for $O(n^2)$-subdivision. The $k$-planar,
 $k$-gap-planar, and fan-crossing  graphs are not universal
 for node-to-circle expansion and for $f(n)$-subdivision if $f(n) \in o(n^2)$ and $f(n) \in o(n   / \log^2n)$
for fan-crossing graphs.
\item
The fan-crossing free, quasi-planar, $k$-gap-planar ($k \geq 1$) and RAC graphs are closed under
path-addition, whereas the fan-crossing graphs and the $j$-planar
graphs ($j \leq 2$) are not.
\item
For every $k \geq 0$ there are fan-crossing free (and even RAC)
graphs  that are not fan-crossing,   $k$-planar, and   $k$-gap
planar, respectively.
\item
The cliques $K_5$ and $K_6$   and some other graphs have a
  unique fan-crossing free embedding.
  \item There are thinned maximal   fan-crossing free graphs with only
 $3.5n-8.5$ edges.
 \item Recognizing fan-crossing free graphs is NP-complete.
\end{enumerate}

The term ``beyond-planar'' is commonly used for a collection of
graph classes $\mathcal{G}$ that extend the planar graphs and are
defined by specific properties of crossings in drawings
\cite{dlm-survey-beyond-19}. These classes have some properties in
common, such as
\begin{enumerate}
  \item [(i)] every planar graph is in $\mathcal{G}$
  \item [(ii)]the density of graphs in $\mathcal{G}$ is   linear
    (up to a poly-logarithmic factor)
  \item [(iii)] $\mathcal{G}$ is universal for subdivision.
\end{enumerate}

These properties are fulfilled by the aforementioned ``major''
beyond-planar graph classes and by almost all graph classes listed
in \cite{dlm-survey-beyond-19}. Graphs that are defined by
generalized visibility  representations, such as bar-$k$
\cite{ekllmw-b1vg-14} and shape visibility \cite{b-Tshape-18,
hsv-rstg-99, lm-Lvis-16} satisfy (i)-(iii), as well as
 graphs with bounded
thickness \cite{deh-gtcg-00, hsv-rstg-99}, bounded stack number (or
book thickness $k$) \cite{bk-btg-79, dw-sqt-05, y-epg4p-89}
 bounded queue number
\cite{dw-sqt-05, djmmuw-queue-20},  mixed stack queue graphs \cite{dw-sqt-05},
and multiple deque layouts \cite{abbbg-deque-18}. In all these
cases, the graphs have a visual representation, which may use an
edge coloring and a linear vertex ordering.

The  properties  exclude graphs with bounded genus and minor-closed
graph classes \cite{d-gt-00}, since they fail on (iii), whereas the
sets of all  graphs and of all non-planar graphs are excluded by
(ii). From the collection of graph classes listed in the
beyond-planar survey \cite{dlm-survey-beyond-19},  skewness-$k$ and
apex-$k$-graphs are excluded, which are such that the removal of at
most $k$ edges and vertices, respectively, makes a
planar graph. \\

\begin{table}[t]
\centering
\begin{tabular}{ l |       c |       c     |   c              |   c          | c
|   c}
             &               &      closed    & $s$-subdiv.       &    closed      &  n2c    &  closed\\
 graph class &   density   &     subdiv.   & univ.                 &   n2c         & univ.   & path-add.\\
  \hline
  planar       & $3n-6$         &     +   &  --                    &     +          & \, --  & --\\
  1-planar     & $4n-8$         &     +   &  $\Theta(n^2)$         &     +          & \, --  & --\\
   1-gap-planar & $5n-10$       &     +   &  $\Theta(n^2)$         &     +          & \, --  & +\\
  fan-crossing & $5n-10$       &      +   &  $\Omega(n/ \log^2n)$ &     ?          & \, --  & --\\
fan-crossing free & $4n-8$     &      +   &  $\leq 2$              &     +          & \, +   & +\\
quasi-planar & $6.5n-20$       &      +    &     1                 &     +          & \, +   & + \\
RAC          & $4n-8$          &      +    &     3                 &     +          & \, ?  & +\\
  \hline
\end{tabular}
 \caption{Properties of some classes of beyond-planar graphs, their
 density, closure under subdivision, node-to-circle  expansion (n2c),
 and path-addition, respectively, and universality for
 subdivision and node-to-circle   expansion.
 Here   ``+'' means yes, ``$-$'' means no, and ``$\leq s$'' says that
 $s$-subdivisions suffice and  $s-1$-subdivision is open. We conjecture ``$-$'' for ``?''..
 }
  \label{table1}
\end{table}

The closure properties of major classes of beyond-planar graphs are
summarized in Table~\ref{table1}. Inclusion relations   and
the containment of special graphs are displayed in
Fig.~\ref{fig:inclusion}.\\

\begin{figure}[t]
  \centering
    \includegraphics[scale=0.7]{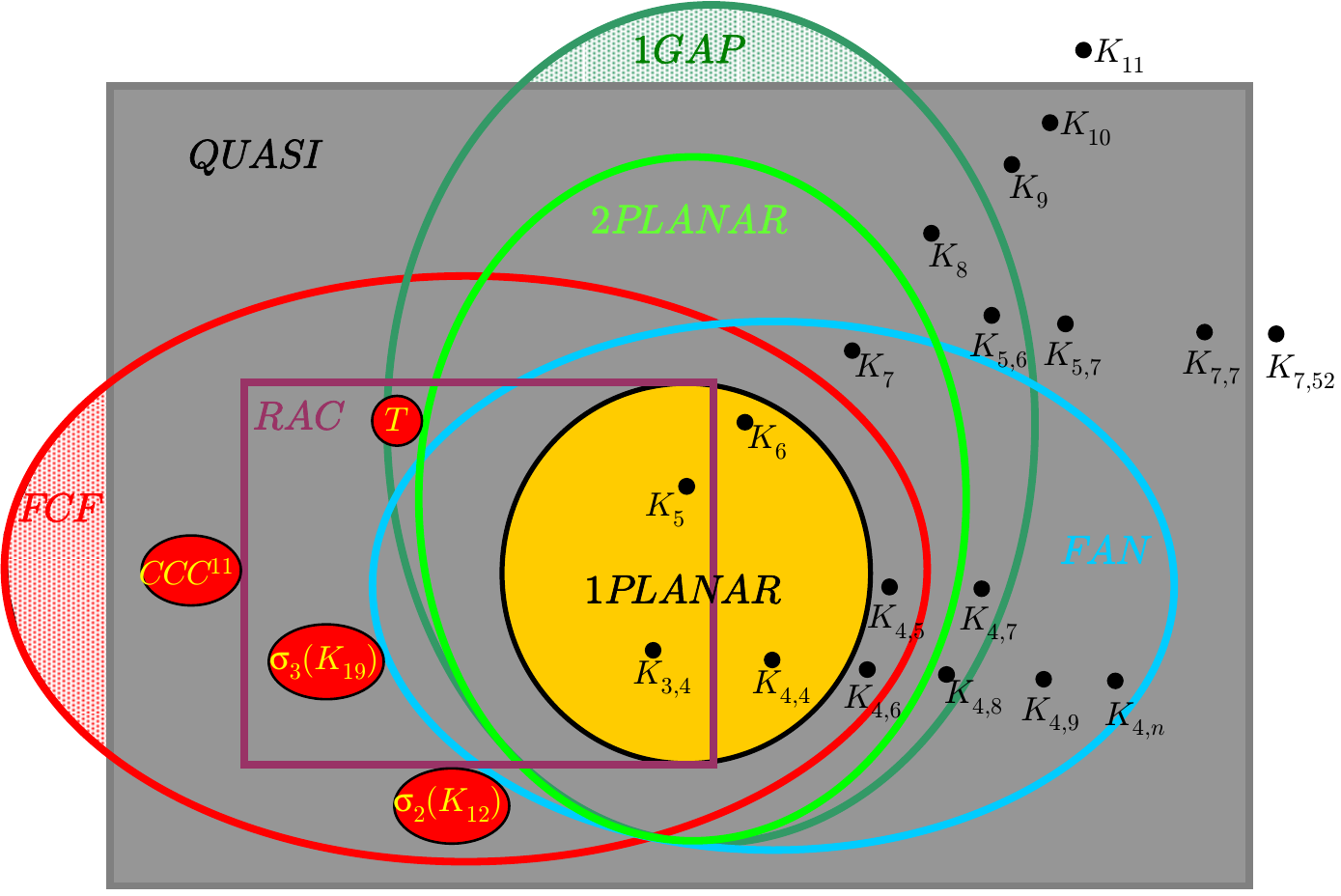}
  \caption{
  The inclusion diagram shows proper inclusions and
  incomparabilities among major beyond-planar graph classes,
  and the containment  of  complete (bipartite) graphs \cite{abks-beyond-Kn-19}.
  It is open, whether there are graphs in \textsf{FCF} and
  1\textsf{GAP}, respectively, that are not in \textsf{QUASI}.
  We contribute cube-connected cycles, in particular $CCC^{11}$,
  2- and 3-subdivisions of complete graphs, namely $\sigma_2(K_{12})$
  and $\sigma_3(K_{19})$, and the tiles-graph $T$
  (Thm.~\ref{thm:tiles}). A graph is not contained in a
  graph class if it is outside its boundary. It may be inside (unlikely) if it
  touches the boundary, such as  graph $T$ for 2\textsf{PLANAR}  and 1\textsf{GAP}.
}
 \label{fig:inclusion}
\end{figure}

The rest of the paper is structured as follows: In
Section~\ref{sect:operations}, we study three graph operations and
introduce   universality of a graph class for an operation. In
Section~\ref{sect:inclusion}, we establish  incomparabilities
between some classes of beyond-planar graphs that are unknown so
far. Finally, in Section~\ref{sect:properties}, we study properties
of fan-crossing free graphs and show that some graphs have a unique
fan-crossing free embedding, that there are thinned  maximal
fan-crossing free graphs, and that the recognition problem is
NP-complete.

\section{Graph Operations and Universality} \label{sect:operations}

Planar subgraphs and   planarization play an important role for the
study of  many classes of beyond-planar graphs
\cite{dlm-survey-beyond-19}. This may indicate that beyond-planar
graphs are close to planar graphs, i.e., they are ``nearly planar''
\cite{ekllmw-b1vg-14}. A major distinction comes from subdivisions.
Clearly, every graph $G$ has a subdivision   that   is 1-planar. On
the other hand, a graph is planar if and only if its subdivision is
planar.

\begin{definition}
A class of graphs $\mathcal{G}$ is \emph{universal} for a graph
operation $f$, or simply $f$-universal,  if for every graph $G$
there is a graph $G' \in \mathcal{G}$ such that $G' = f(G)$.
\end{definition}

Clearly, the removal of all edges transforms any graph into a
discrete graph, so that the set of discrete graphs is universal for
edge removal. Similarly, the set of complete graphs is universal for
edge insertion. A graph class $\mathcal{G}'$  is universal for $f'$
if $\mathcal{G} \subseteq \mathcal{G}$',  $f'$ extends $f$, and
$\mathcal{G}$ is universal for $f$.

Universality provides a new perspective on the set of   graphs.
Every  graph is   projected into a special class $\mathcal{G}$ by an
operation $f$, so that \textsf{GRAPHS} = $f^{-1}(\mathcal{G})$,
where \textsf{GRAPHS} denotes the set of   simple graphs. Thus all
graphs look like graphs from
$\mathcal{G}$ if they are seen through the lens of  $f$. \\

For our further studies, box-visibility representations of graphs
and the   number of crossings will be useful.
 In a \emph{box-visibility representation}, the
vertices of a graph are represented as  boxes, which are axis
aligned rectangles. Each edge is an orthogonal polyline with
horizontal and vertical segments between the borders of two boxes,
so that segments do not traverse other boxes. Edges incident to the
same box do not cross. There is a \emph{rectangle visibility
representation} \cite{hsv-rstg-99} if each edge  consists of a
 horizontal or a vertical segment, and a \emph{special box-visibility
representation} if each edge consists of a vertical and a horizontal
segment with a  bend to the right in between. It is well-known that
every graph has a special box-visibility representation. The idea
dates back to the early 1980th and Valiant's representation of
graphs of degree at most four \cite{v-ucVLSI-81}. Simply, place the
boxes for the vertices on the main diagonal and route each edge
$\edge{u}{v}$ by a vertical segment from the top of the  box of $u$
and a horizontal segment to the left side of the box of $v$. Later
studies have focussed on box-visibility representation with small
area and few bends \cite{bk-aesigd-97, bmt-3pm-00,pt-aeod-98}.
%
%
Box-visibility representations  have been used to show that every
graph has a 3-bend RAC drawing \cite{del-dgrac-11}. There is an
alternative 3-bend RAC drawing, in which the vertices are placed on
a horizontal line and the second and third segments are drawn with
slope $\pm 1$.\\

 Ajtai et al.~\cite{acns-cn-82} and, independently, Leighton \cite{l-VLSI-83} discovered
 that the \emph{crossing number} $cr(G)$  of any graph $G$ with $n$ vertices and $m > 4n$
 edges is at least $cm^3/n^2$ for some   $c>0$. This fact is known as \emph{Crossing Lemma}.
 The currently best bound is $c=1/29$
if $m \geq 6.95 n$  \cite{a-cn-19}. In consequence, complete
 graphs $K_n$ have $\Omega(n^4)$ many crossings. The Gay or Harary-Hill
 conjecture states  that $cr(K_n) = \frac{1}{4} \lfloor
 \frac{n}{2} \rfloor \lfloor  \frac{n-1}{2} \rfloor \lfloor  \frac{n-2}{2} \rfloor \lfloor
 \frac{n-3}{2} \rfloor$, which has been proven for $n \leq 12$
 \cite{pr-crK11-07}.
 Hence, any drawing of $K_{12}$ has at least 150 crossings.
 Similarly, there are upper and lower bounds on the
 crossing numbers of complete bipartite graphs \cite{k-cn-70},
 hypercubes and cube-connected cycles, where
 a lower bound is $\frac{1}{20} 4^d - (9d+1)2^{d-1}$ for the
 $d$-dimensional cube-connected cycle \cite{sv-cnCCC-93}.

\subsection{Node-to-Circle Expansion}

A \emph{node-to-circle expansion} substitutes each vertex $v$ of
degree $d$  by   a  circle $C(v)$ of length $d$, so that each vertex
has degree three and each edge incident to $v$ is inherited by a
vertex of the circle. If there is a rotation system with the cyclic
ordering of the edges incident to $v$, then this ordering is
preserved by the node-to-circle expansion. Thereby edges do not
cross if they are incident to vertices of a circle, since graphs are
simple. There are \emph{inner edges}  between consecutive vertices
of a circle and \emph{binding edges} between vertices of different
circles, which are one-to-one related to the original edges.
A  \emph{node-to-box expansion} is the special case, in which the
edges of each circle are uncrossed.
We denote the node-to-circle expansion of a graph $G$ by $\eta(G)$,
which  is a 3-regular (cubic) graph with $2m$ vertices and $3m$
edges   if $G$ has $m$ edges and no vertices of degree one or two.

 A node-to-circle expansion is a special split  operation, in
which each vertex is replaced by a subgraph $H$, so that the
vertices of $H$ inherit all adjacencies from the vertex. In our
case, $H$ is a circle. Graph $H$ is a discrete graph in the
$k$-split operation from \cite{EKLLLMMVWW-split-18}, where the
objective is to transform a graph into a planar one by as few
$k$-split operations as possible.

Node-to-circle expansions   have been used in VLSI theory to
transform a hypercube  into a cube-connected cycle \cite{l-paa-92}.
The $d$-dimensional hypercube $H^d$  consists of $2^d$ vertices
denoted by $d$-digit binary numbers. There is an edge if and only if
the Hamming distance between the binary numbers is one. The
$d$-dimensional cube-connected cycle $CCC^d$ is obtained from $H^d$
by beveling each corner, so that  $CCC^d$ is the node-to circle
expansion of $H^d$, see Fig.~\ref{fig:CCC4}.
 It has $d2^d$ vertices of degree three.
Hypercubes and cube-connected cycles have similar properties, such
as the crossing number \cite{sv-cnCCC-93}, diameter and  separation
width \cite{l-paa-92}.

The node-to-circle expansion operation preserves some graph
properties, such as 3-connectivity and  genus, and simplifies
others, such as  crossings. If $G$ has crossing number $k$, then
 $\eta(G)$ has crossing number at
most $k$. For example, $K_{10}$ has 60 crossings and
  is not 4-planar \cite{abks-beyond-Kn-19}, whereas
 $\eta(K_{10})$ has at most 45 crossings and is 4-planar,
as Fig.~\ref{fig:K10-n2c} shows. The complete bipartite graphs
$K_{4,n}$ are fan-crossing but are not $k$-planar for $k=1,2,3,4$
and $n=5,7,10,19$, respectively \cite{abks-beyond-Kn-19}. Moreover,
graph $K_{4,7}$ is not fan-crossing free. However, the
node-to-circle expansion of $K_{4,n}$ is 1-planar, as
Fig.~\ref{fig:K4n-node-to-circle} illustrates, and it has fewer
crossings than $K_{4,n}$.
 Clearly, the crossing numbers of $G$ and $\eta(G)$
coincide if the inner edges of $\eta(G)$ are uncrossed, i.e. in a
node-to-box expansion.

\begin{figure}[t]
  \centering
  \subfigure[]{
    \includegraphics[scale=0.7]{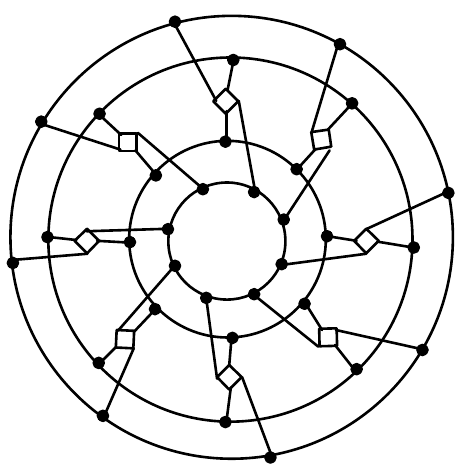}
    \label{fig:K4n-node-to-circle}
    }
    \hspace{2mm}
 \subfigure[]{
        \includegraphics[scale=0.5]{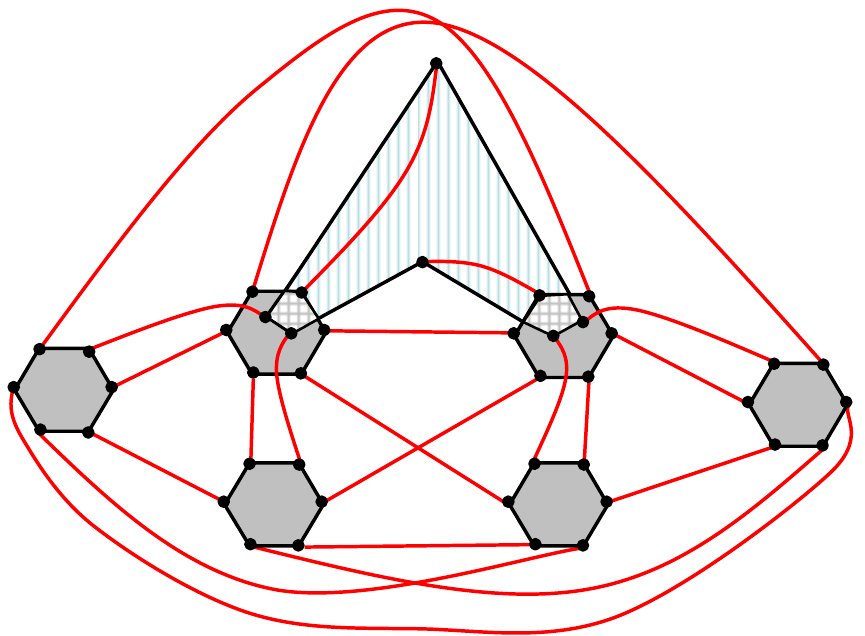}
\label{fig:K7-node-to-circle}
        }
\caption{A 1-planar drawing of the node-to-circle expansion of (a)
$K_{4,n}$ and (b) $K_7$.
In (b) 
the topmost circle (striped) intersects two other circles.
  }
  \label{fig:nodeexp}
\end{figure}

\begin{figure}[t]
  \centering
  \subfigure[]{
    \includegraphics[scale=0.2]{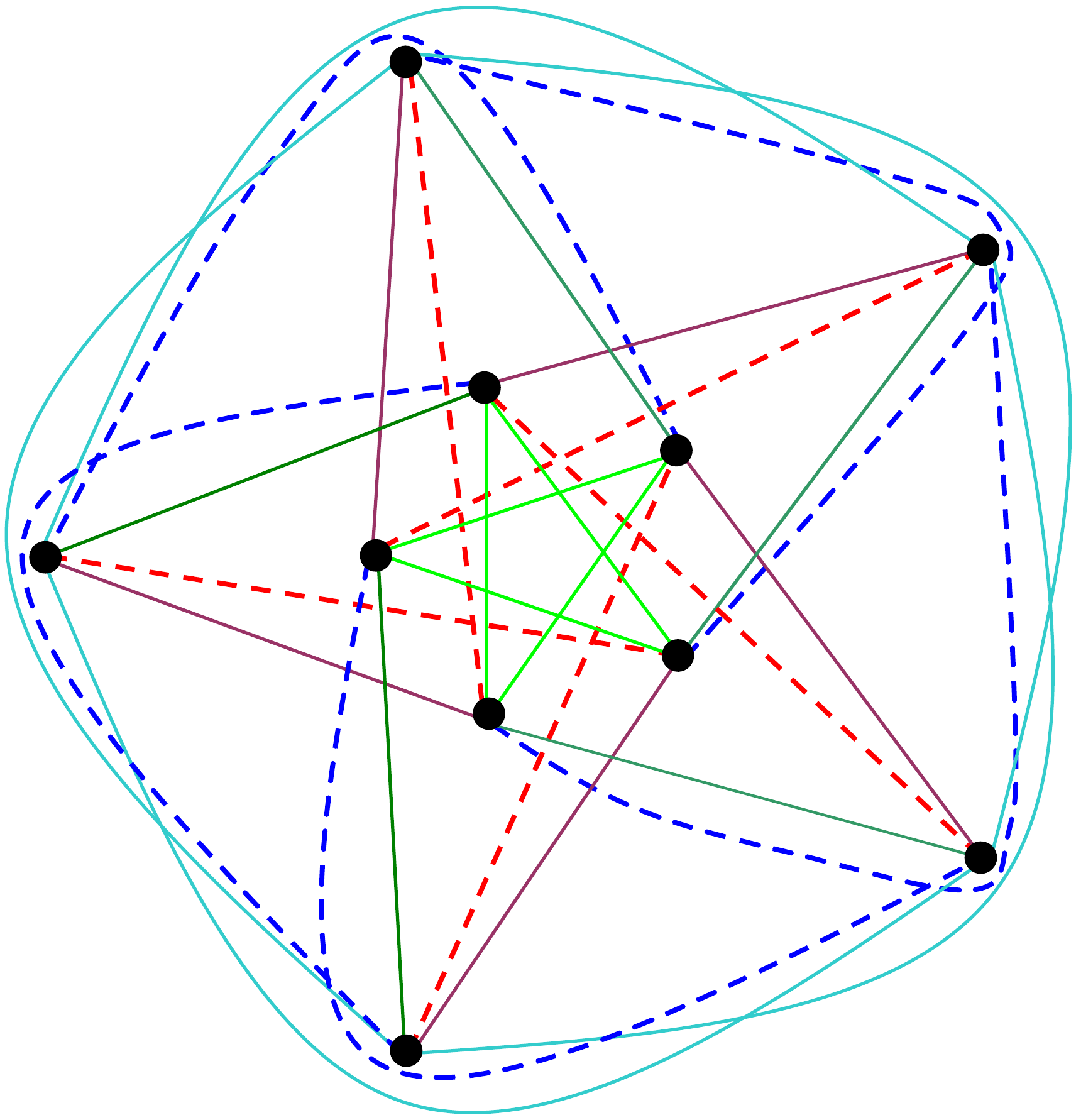}
    \label{fig:K10-g}
    }
    \hspace{4mm}
 \subfigure[]{
        \includegraphics[scale=0.2]{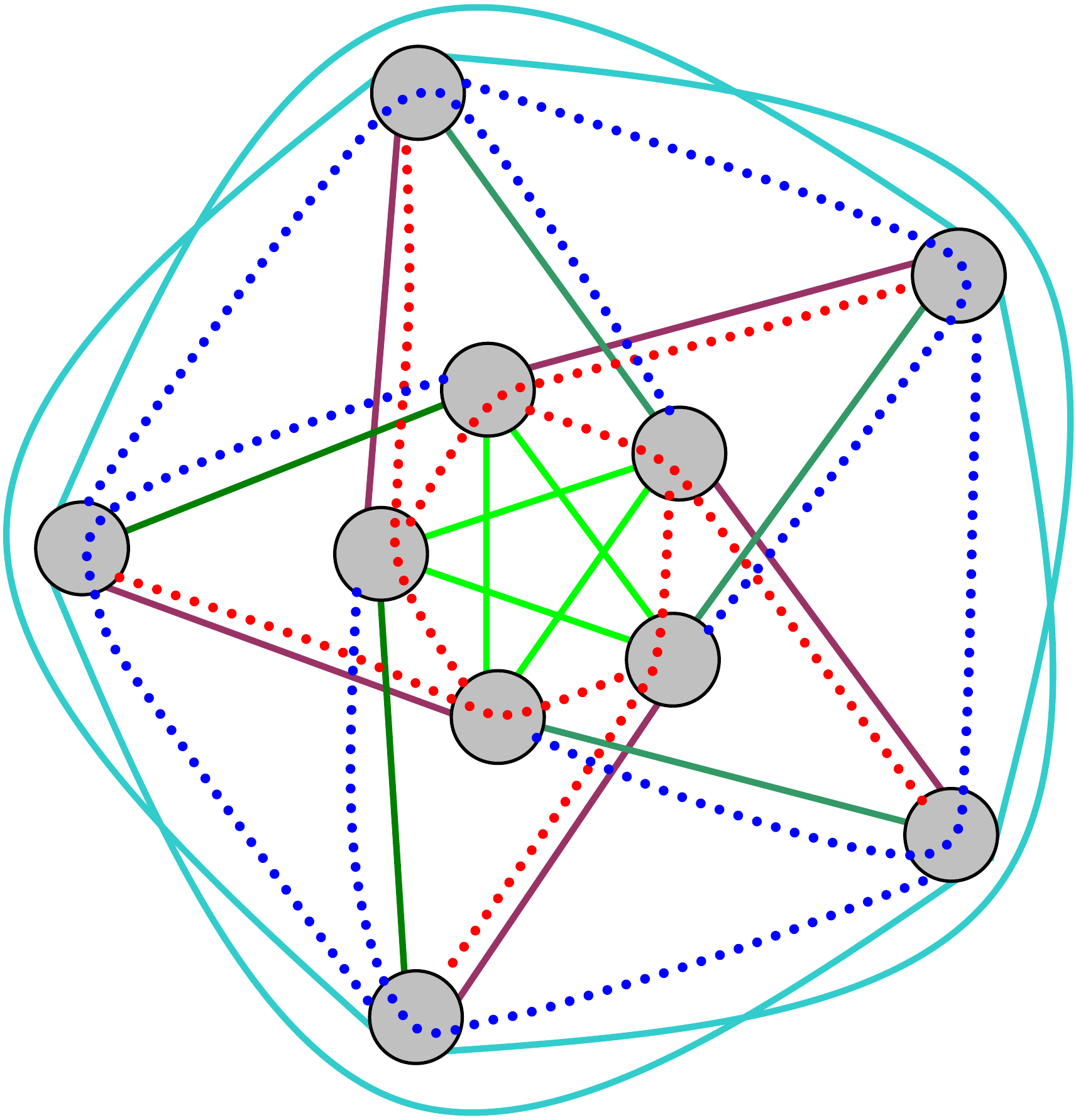}
\label{fig:K10-n2c}
        }
\caption{Quasi-planar drawings of (a) $K_{10}$ and (b)
$\eta(K_{10})$, where edges from $K_{10}$, drawn dashed and red and
blue, are drawn dotted and red and blue in $\eta(K_{10})$. These
 edges are crossed six times in   $K_{10}$, whereas the binding edges in $\eta(K_{10})$ are
crossed at most four times, including two crossings on the boundary
of the traversed circle, drawn as a big node. }
  \label{fig:K10}
\end{figure}

\begin{lemma} \label{lem:closure-node-to-circle expansion}
The node-to-circle expansion $\eta(G)$ of a graph $G$ is   planar if
and only if $G$ is planar. Similarly, $\eta(G)$ is fan-crossing free
($k$-planar, $k$-gap-planar, quasi-planar, RAC) if so is $G$.
\end{lemma}
\begin{proof}
Construct graph $\eta(G)$ from a drawing of $G$. Consider a small
box    around each vertex, so that
 two boxes do not intersect and a box is not crossed by a
non-incident edge.
  For each edge $e$ incident to
$v$ place a new vertex $v_e$ at the last intersection of $e$ and the
box,  clip $e$ at $v_e$ and remove $v$. Link two vertices by an edge
that are adjacent in the boundary of a  box, so that each box is
transformed into a circle with uncrossed edges.
 For each edge $e=\edge{u}{v}$ of $G$ there is an
edge $e'$ between the boxes or circles of $u$ and $v$, so that $e'$
inherits all crossings from $e$. Hence, $\eta(G)$ is fan-crossing
free ($k$-planar, $k$-gap-planar, quasi-planar, RAC) if so is $G$.
\end{proof}

\begin{corollary} \label{cor:closure-node-to-circle expansion}
The graph classes \textsf{FCF}, $k$\textsf{PLANAR}, $k$\textsf{GAP},
\textsf{QUASI}, and \textsf{RAC} are closed under node-to-circle
expansion ($k\geq 0$).
\end{corollary}

Angelini et al.~\cite{abks-ldgbp-18} observed that   every
fan-planar drawing of a degree-3 graph is   3-planar. This fact can
be improved. Here  configuration II comes into play and  edges are
rerouted as in \cite{b-fan-20}.

\begin{lemma} \label{lem:fan-2planar}
Graph $G$ is 2-planar and thus 1-gap-planar if it is a
fan-crossing graph of degree at most three.
 \end{lemma}
 \begin{proof}
Consider a fan-crossing drawing of $G$. If edge $e$ is crossed by
three edges, then the crossing edges are incident to a vertex $v$
and $e$ is not crossed by any other edge. Now vertex $v$ can be
moved close to the crossing point of $e$ and the middle of the three
crossing edges. Similar to \cite{b-fan-20}, the edges incident to
$v$ are rerouted first along $e$. Then $e$ is uncrossed if it is
crossed from one side, as in the fan-planar  case
\cite{ku-dfang-14}, and it is crossed at most once, in general, see
Fig.~\ref{fig:fanplanar}. Hence, there is a 2-planar drawing, which
is 1-gap-planar \cite{bbc-1gap-18}.
%
\end{proof}

\begin{figure}[t]
  \centering
  \subfigure[]{
    \includegraphics[scale=0.9]{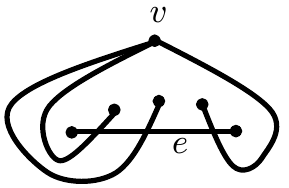}
    }
    \hspace{2mm}
 \subfigure[]{
        \includegraphics[scale=0.9]{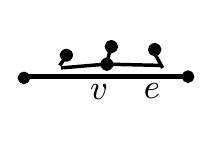}
   }
      \hspace{2mm}
 \subfigure[]{
        \includegraphics[scale=0.9]{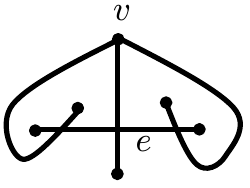}
   }
      \hspace{2mm}
 \subfigure[]{
        \includegraphics[scale=0.9]{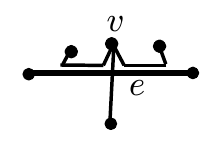}
   }
\caption{An edge $e$ that is crossed by three edges incident to
vertex $v$ in a 3-regular fan-crossing graph. Edge $e$ is crossed
from (a)   one  and (c)  both sides (configuration II).  A new
placement of $v$ and an edge rerouting towards a 2-planar graph. }
  \label{fig:fanplanar}
\end{figure}

In particular, the node-to-circle expansion $\eta(G)$ is 2-planar
and 1-gap-planar if it is fan-crossing.
It is unclear whether $\eta(G)$ is fan-crossing if $G$ is
fan-crossing. It is true for maximal complete (bipartite) fan-crossings graphs,
that is $K_7$ and $K_{4,n}$,  where the node-to circle
expansions are even 1-planar, as Fig.~\ref{fig:nodeexp} shows. These
cases are singular. Consider a $4\times 4$ grid graph with nine
inner quadrangular faces. Place $n \geq 5$ vertices in each
quadrangle and add edges so that there is $K_{4,n}$ including the
four corners of each quadrangle. We conjecture that the
node-to-circle of this graph is not fan-crossing.

\begin{theorem} \label{thm:node-expansion-uni}
(i) The   fan-crossing-free graphs and the quasi-planar graphs are
universal for node-to-circle expansion.

(ii) The fan-crossing, $k$-planar and  $k$-gap-planar graphs ($k
\geq 0$) are not universal for node-to-circle expansion.
\end{theorem}

\begin{proof}
For a fan-crossing free drawing of $\eta(G)$, consider a  drawing of
$G$  with all vertices on a circle in the outer face. Replace each
vertex $v$ by a circle $C(v)$  that does not intersect any other
circle or an edge that is not incident to $v$, so that all inner
edges are uncrossed. If two edges of $G$ cross, then so do the
corresponding binding edges, which are independent, so that
$\eta(G)$ is fan-crossing free.

For quasi-planar graphs, consider a
special box-visibility representation of $G$, in which   the boxes
are arranged along the diagonal. Each edge consists of a vertical
and a horizontal segment with a bend to the right. Direct it upward
so that each box is incident to incoming edges on the left
 and of outgoing edges on top, see Fig.~\ref{fig:boxvisK6}.
 For each vertex $v$ construct a
circle $C(v)$ consisting of the bend points of the outgoing edges
and the points of the incoming edges at the left side of the box of
$v$, as shown in Fig.~\ref{fig:box2circle}. The ordering  of $C(v)$
coincides with the rotation system at the box of $v$. The inner
edges from the cycles are almost vertical. Two such edges do not
cross, since they are ordered left to right according to the left to
right ordering of the boxes. The binding edges are the horizontal
segments of the edges from the special box-visibility
representation. Such edges may cross inner edges. Nevertheless, there are
no three mutually crossing edges, so that $\eta(G)$ is quasi-planar.

For (ii) consider  hypercubes and cube-connected cycles as their
node-to-circle expansion. The crossing number of $CCC^d$ is at least
$\frac{1}{20} 4^d - (9d+1)2^{d-1}$ \cite{sv-cnCCC-93}, so that some
edges are crossed at least $c2^d/d$ times for some $c>0$, see
\cite{abks-ldgbp-18}. Hence, for every $k \geq 0$ there are
hypercubes so that their node-to-circle expansion
 is not $k$-(gap)-planar and not fan-crossing by
Lemma~\ref{lem:fan-2planar}.
 \end{proof}

 \begin{figure}[t]
  \centering
  \subfigure[]{
    \includegraphics[scale=0.37]{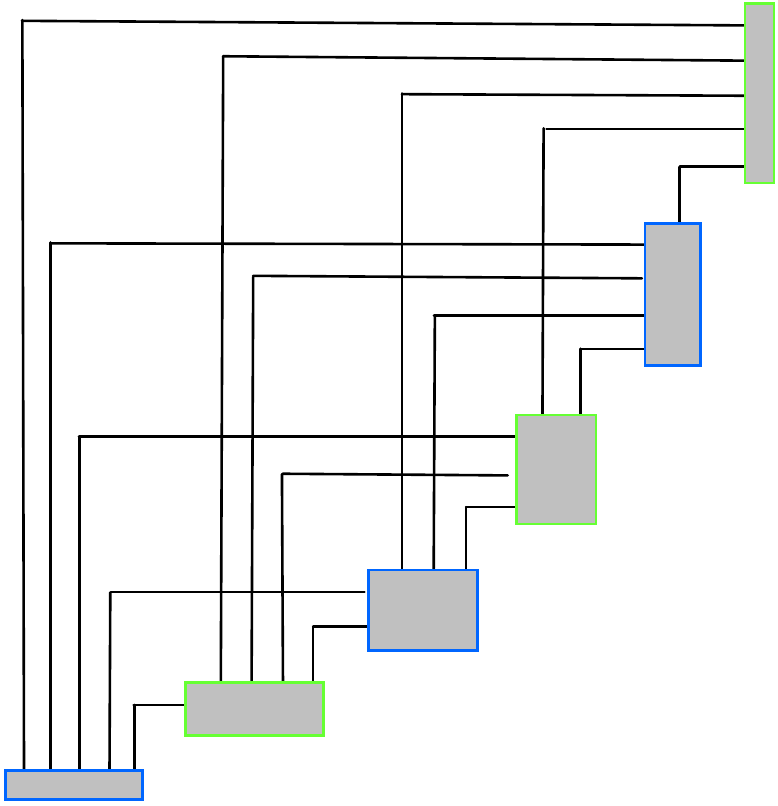}
    \label{fig:boxvisK6}
    }
    \hspace{5mm}
 \subfigure[]{
        \includegraphics[scale=0.4]{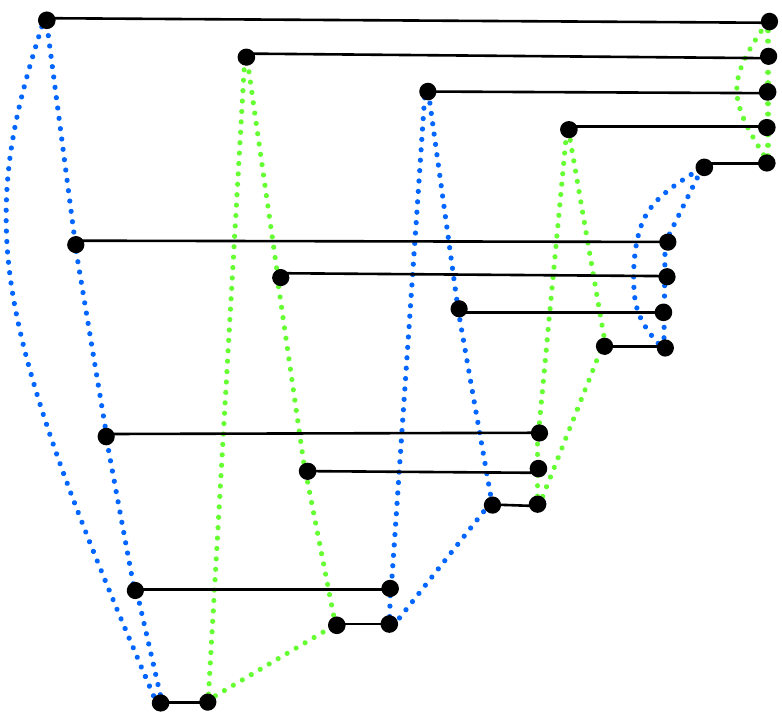}
\label{fig:box2circle}
        }
\caption{(a) A special box-visibility representation of $K_6$ and
(b) its node-to-circle expansion. Inner edges are drawn dotted.}
  \label{fig:boxvis}
\end{figure}

In consequence, the 11-dimensional cube-connected cycle $CCC^{11}$
is fan-crossing free and quasi-planar but not 1-gap-planar.
Moreover, it is not 2-planar \cite{bbc-1gap-18} and not
fan-crossing by Lemma~\ref{lem:fan-2planar}.

Note that the fan-crossing free and quasi-planar  drawings of
$\eta(G)$ are different. They are not necessarily quasi-planar and
fan-crossing free, respectively. It is open whether every graph $G$
has a drawing  of $\eta(G)$  that is simultaneously quasi-planar and
fan-crossing free, called grid-crossing in \cite{b-FOL-18}, or that
is RAC. A RAC drawing can be obtained by a \emph{generalized
version} of a node-to-cycle expansion  in which  a  circle $C(v)$ of
a degree-$d$ vertex  has $d$ vertices of degree three and  two
vertices of degree two. Then $C(v)$ can be drawn as a box with
vertical and two horizontal  segments. Moreover, every graph has a
generalized node-to-circle expansion that is outer fan-crossing
free, so that all vertices are in the outer face.

\subsection{Subdivision}

An $s$-subdivision of an edge replaces it by a path of length $s+1$.
In geometry and in particular in orthogonal \cite{efk-ogd-99} and
RAC drawings \cite{del-dgrac-11}, a subdivision is called a bend.
Each edge is replaced by a path of length at most $s+1$ in   an
$s$-subdivision $\sigma_s(G)$ of graph $G$, where $s$ is the number of edge
for a subdivison of $G$. An $s$-subdivision is
  \emph{uniform} if  each edge of graph $G$
is extended to a path  of length $s+1$. Then  $\sigma_s(G)$ is a
single graph. In the general version, $\sigma_s(G)$ is parameterized
by its set of edges and is a set of graphs.  Here these cases do not
 matter, since too long paths can be folded, so that some of its
edges remain uncrossed. Subdivisions of graphs are used in
Kuratowski's theorem, which states   that a planar graph does not
contain a subgraph that is a  subdivision of $K_5$ or
$K_{3,3}$, i.e., $K_5$ and $K_{3,3}$ are the topological minors  of
the planar graphs. Note that edge contractions, i.e., paths of
length zero, are generally used in the theory of graph minors
\cite{d-gt-00}.

Clearly, the ($s$-)subdivision of a graph $G$ is in a graph class
$\mathcal{G}$ if $G$ is   one of the aforementioned graph classes.
If edges of a fan with a common vertex $v$ are crossed, then their
 segments at  $v$ are crossed after a subdivision.

\begin{lemma} \label{lem:closure-subdivision}
The graph classes \textsf{PLANAR}, $k$\textsf{PLANAR},
$k$\textsf{GAP}, \textsf{FCF}, \textsf{FAN}, \textsf{QUASI}, and
\textsf{RAC} are closed under $s$-subdivision for every $s \geq 0$.
\end{lemma}


 A graph is
\emph{IC-planar} \cite{a-cnircn-08} if it admits a 1-planar drawing
so that each vertex is incident to at most one crossed edge.
Structural properties have been studied in \cite{bbhnr-NIC-17}.
Brandenburg et al.~\cite{bdeklm-IC-16} have shown that IC-planar
graphs admit a 1-planar straight-line drawing with right angle
crossings. Hence, IC-planar graphs are simultaneously 1-planar and
RAC.   Hence, a universality of the IC-planar graphs transfers to
all graph classes containing them.

\begin{theorem} \label{thm:universal}
\begin{enumerate}
  \item  [(i)] The IC-planar, $k$-planar, and $k$-gap-planar  ($k \geq 1$)
  graphs are universal
 for $f(n)$ subdivision if and only if $f(n) \in \Theta(n^2)$. The
 fan-crossing graphs are universal for $f(n)$-subdivision if $f(n) \in O(n^2)$.
 They are not $g(n)$-subdivision  universal if $g(n) \in o(n /
 \log^2 n)$.
  \item  [(ii)] The RAC graphs are universal for 3-subdivision and not for 2-subdivision.
  \item  [(iii)] The fan-crossing free graphs are universal for 2-subdivision.
  \item [(iv)] The quasi-planar graphs are universal for
  1-subdivision and not without subdivisions.
\end{enumerate}
\end{theorem}

\begin{proof}
Consider a   drawing of   graph $G$. Subdivide each edge into
segments, so that each segment  is crossed at most once and each
subdivision point is incident to at most one crossed segment.
The so obtained drawing  is  IC-planar.  Each edge of $G$ is
crossed at most by all other edges,   since $G$ is simple. Thus $O(n^2)$ subdivisions per edge
suffice. For the lower bound consider the crossing number of
complete graphs $K_n$, which is $\Omega(n^4)$ \cite{acns-cn-82,
l-VLSI-83}. Clearly, the number of crossings is unchanged by
subdivisions. Since $k$-planar and $k$-gap-planar graphs admit
$O(m)$ crossings, where $m$ is the number of edges, some edges have
$\Omega(n^2)$  crossings and need $\Omega(n^2)$ many subdivisions
for   $k$-(gap)-planar graphs. Clearly, $O(n^2)$ subdivisions
suffice for fan-crossing graphs.  For the lower bound consider
$d$-dimensional cube-connected cycles. By the lower bound on the
crossing number \cite{sv-cnCCC-93},   there are edges with at least
$\frac{1}{30} \frac{2^d}{d}-4$ many crossings, which is $c
\frac{n}{\log^2 n}$ for some $c>0$, since $n=d2^d$. By Lemma
\ref{lem:fan-2planar}, 3-regular fan-crossing graphs are 2-planar,
so that at least  $cn / \log^2 n$ many subdivisions are needed for
some edges.

  Didimo et al.~\cite{del-dgrac-11} have shown that
every graph has a 3-bend RAC drawing, whereas $n$-vertex graphs with
a 2-bend RAC drawing have   $O(n^{7/4})$  edges, which proves (ii).

For (iii), every graph can be drawn such that each edge consists of
three segments, the first and last of which are uncrossed and the
middle segments inherit all crossings. Hence, the drawing is
fan-crossing free and has two subdivisions.

Finally, for (iv), consider a special box-visibility representation
of $G$. Use each bend as a subdivision point and draw each vertex
$v$ as a point in its box.   The incident edges of $v$ are extended
from the boundary of the box to $v$ so that there is no crossing in
the interior of a box. Then only horizontal and (almost) vertical
segments of edges may cross,
  so that no three edges cross one
another. Since $n$-vertex quasi-planar graphs have at most $6.5n-20$
edges \cite{at-mneqpg-07}, one subdivision is necessary to represent
all graphs.
%
\end{proof}

\subsection{Path-Addition}

A \emph{path-addition} adds an internally vertex-disjoint path $P$
between any two vertices of a graph. Such paths are also known as
ears. They are used in ear decompositions of 2-connected graphs
\cite{gy-graphtheory-99} and in  subdivisions. We wish to use
path-additions so that they preserve a given class of graphs.
Therefore, the added paths are long. Their  length is at least $cn$
for some $c > 0$ and $n$-vertex graphs. This property distinguishes
our path-addition operation from  ear decompositions
\cite{gy-graphtheory-99}. In comparison with subdivision, a
path-addition adds a (long) path between any two
vertices, whereas a subdivision can do so if there is an edge. Hence,
path-additions can create a subdivision of $K_n$ from any non-empty graph. A
more general version of path-addition  was introduced in
\cite{ben-ab1v-17} and studied in \cite{ben-pa-16}. For a graph $G$
and a path $P$, let $G \oplus P$ denote the graph obtained by adding
the vertices and edges of $P$ to $G$, where the internal vertices of
$P$ are new and have degree two. A class of graphs $\mathcal{G}$ is
closed under path-addition if there is some function
$D_{\mathcal{G}}(n)$ so that $G \oplus P$ is in $\mathcal{G}$ if $G$
is an $n$-vertex graph in $\mathcal{G}$ and $P$ has length at least
$D_{\mathcal{G}}(n)$.

It is readily seen   that the graph classes \textsf{RAC},
\textsf{FCF}, $k$\textsf{GAP} for $k \geq 1$, and \textsf{QUASI} are
closed under path-addition if $D_{\mathcal{G}}(n)$ is a linear
function. In particular, $k$-gap-planar graphs are made for path-additions,
since the crossings created by an added paths are assigned to its edges
if the length of the path exceeds its number of crossings.
For a non-closure it must be shown that there are graphs so that the
addition of a path of any length violates the defining properties of
the graph class.

\begin{theorem} \label{thm:path-addition} 
The graph classes \textsf{RAC}, \textsf{FCF}, $k$\textsf{GAP} ($k
\geq 1$), and \textsf{QUASI} are closed under path-addition, whereas
  $k$\textsf{PLANAR} ($k \leq 2$) and \textsf{FAN}
are not. 
\end{theorem}

\begin{proof}
For the positive closure results, let $D_{\mathcal{G}}(n)$ be twice
the density function.  Path $P$ is routed along a (shortest) path
$S$ in $G$.
 Path $P$ makes a bend and creates a further subdivision if it
crosses an edge incident to an internal vertex of $S$, so that the
crossing introduces a violation without the subdivision.  In a RAC
drawing,  the edges of  $P$ can be drawn  so that they
  cross an edge of $G$ at a right angle. If a graph
is disconnected and the first and last vertices of $P$  are in
different components, then there are exits, which are vertices or
crossing points in the outer face of a component, so that $P$ can be
routed along such exists. Each bend is charged to an edge  of the
given graph, which is charged at most twice. If  an
edge of $G$ crosses $P$ multiple times, then it crosses $S$ so that
these crossings are assigned to the edges of $S$.

For the negative results, consider a $4\times 4$ grid graph $G$ with
vertices $(x,y)$ for $1 \leq x,y \leq 4$. Triangulate  $G$ by edges
parallel to the diagonal. Then replace each edge by a ``fat edge''.
For fan-crossing (fan-planar) graphs, a fat edge $\edge{u}{v}$ is
$K_7$ with vertices $u$ and $v$ in the outer face of an embedding.
Call the resulting graph $H$. Binucci et al.~\cite{bddmpst-fan-15}
have shown that $K_7$ has an almost unique fan-planar  embedding
with two vertices in the outer face, which results from fragments of
crossed edges, see Fig.~\ref{fig:K7-fan-2planar}. The result  holds
for fan-crossing graphs, since configuration II \cite{ku-dfang-14}
is not used, and for 2-planar graphs. Fat edges cannot be crossed
without violation. In case of 1-planar graphs, a fat edge is $K_6$,
see Fig.~\ref{fig:K6},  and   an edge for planar graphs. In
consequence, any fan-crossing or 1- or 2-planar embedding of $H$ is
a grid with fat edges. For every fan-crossing ($j$-planar with
$j\leq 2$) drawing of $H$, there is a fat edge between the vertices
$u$ and $v$ at positions $(2,2)$ and $(3,3)$,   which cannot be
crossed without violation. Hence, a path $P$ cannot be added between
$u$ and $v$ so that $H \oplus P$ is fan-crossing ($j$-planar with $j
\leq 2$),    if  $G$ is fan-crossing ($j$-planar).
\end{proof}

 Hence, the path-addition
operation distinguishes the graph classes \textsf{RAC},
\textsf{FCF}, $k$\textsf{GAP} ($k \geq 1$), and \textsf{QUASI}
from \textsf{FAN},
$1$\textsf{PLANAR} and 2\textsf{PLANAR},
respectively. Since every 1-planar graph is fan-crossing free
\cite{cpkk-fan-15}, every 2-planar graph is 1-gap-planar
\cite{bbc-1gap-18}, and every fan-crossing or 2-planar graph \cite{abbddd-quasiplanar-20}
is quasi-planar,
it provides an alternative proof for the
properness of these inclusions.

\section{Relationships} \label{sect:inclusion}

The following proper inclusions are known:   Every $2k$-planar graph
is $k$-gap planar, which in turn is $2k$$+$$2$-quasi-planar
\cite{bbc-1gap-18}. Every RAC graph is fan-crossing free and
quasi-planar \cite{del-dgrac-11}. Every fan-crossing graph is
quasi-planar \cite{cpkk-fan-15}.  In addition, the following
incomparabilities are  known: (i) 1-planar and RAC
\cite{el-racg1p-13} (ii)  2-planar and fan-crossing (fan-planar)
   \cite{bddmpst-fan-15}, and (iii)   fan-crossing and 1-gap-planar, where
   the latter, stated as an open problem in \cite{bbc-1gap-18}, follows
   from $K_8$ and $K_{4, n}$ for $n \geq 9$ \cite{abks-beyond-Kn-19}.

We use the
 path-addition, subdivision, and
node-to-circle expansion operations to construct examples for more
incomparabilities and alternative proofs for proper inclusions.

\begin{figure}[t]
  \centering
    \includegraphics[scale=0.6]{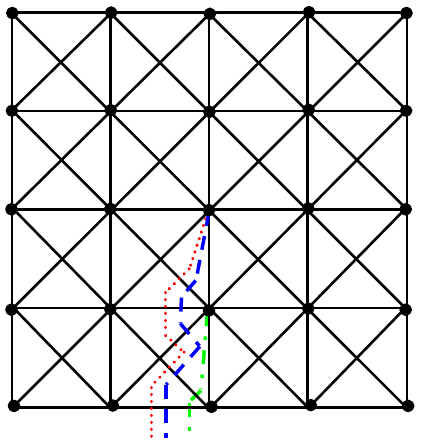}
\caption{An internally crossed $5 \times 5$ grid with three added
paths, drawn dotted and/or dashed (and colored in the online
version) in a RAC drawing. }
   \label{fig:tile}
\end{figure}

\begin{theorem} \label{thm:tiles}
There are RAC   graphs that are not fan-crossing.
\end{theorem}

\begin{proof}
A \emph{tile} $T$ is an internally crossed   $5 \times 5$ grid
  with vertices $t_{p,q}$ for $1
\leq p,q \leq 5$ and edges $\edge{t_{p,q}}{t_{r,s}}$ if
$\max\{|p-r|, |q-s|\} = 1$. 
It has  25 vertices and 72 edges and diameter four, see
 Fig.~\ref{fig:tile}.

 Since outer fan-planar graphs have
a density of $3n-5$ \cite{bddmpst-fan-15}, a tile is not
  outer fan-planar. Note that every outer fan-crossing graph is outer
fan-planar, since an edge cannot be crossed from both sides if all
vertices are in the outer face. Thus there is no configuration II
\cite{ku-dfang-14}. Hence, in every fan-crossing embedding, each
tile $T$ has at least one vertex in its interior, called the
\emph{center} of $T$, and at most 24 edges in the boundary of its
outer face.  There are at least three vertex-disjoint paths between
the center and vertices in the boundary, since a tile is
3-connected.

Let graph $G$ consists of at least  650  tiles and  at least 60
paths between any two vertices from different tiles. Suppose there
is a fan-crossing embedding of $G$. Each tile has a Jordan curve for
its boundary, so that the center is in the interior of the Jordan
curve. We say that two tiles cross, nest, and are disjoint if their
boundaries cross, nest or are disjoint, respectively. They nest if
one tile is in the interior of the
boundary of the other tile. 
Define
the distance    between 
two tiles by the minimum number of boundaries that must be crossed
by a path $P$ between their centers   if $P$ is added to $G$ in a
fan-crossing embedding of $G \oplus P$. In particular, the distance
is at least two if the tiles are disjoint, since an added path must
cross the boundary of each of them.

Suppose there are (at least) five tiles $T_i$   with   mutual
distance at least two. Then consider $p \geq 60$ paths of length 15
between their centers. 
%
Each path must cross at least two boundaries. Since each boundary
consists of at   most 24 edges, at least one of its edges is crossed
twice if there are at least 25 paths. In a fan-crossing embedding,
this edge is crossed by the first (or last) edges of the paths,
since all other edges of the paths are independent. We call  $P$ a
\emph{red path} if the first and last edges cross an edge of a
boundary. Suppose there is a red path between the centers  of tiles
$T_i$ and $T_j$ for $1 \leq i, j \leq 5$. Then there is a subgraph
homeomorphic to $K_5$, so that at least two paths must cross. Since
the paths have length 15 and their first and last edges are crossed
by edges of boundaries, at least one inner edge of a path is crossed
twice by independent edges from two paths if there are at least 14
parallel red paths. Then there is a violation of a fan-crossing
embedding.

Since the centers of the tiles are unknown, there are $p \geq 60$
paths of length 15 between any two vertices from different tiles.
Graph $G$ is not fan-crossing if there are at least five tiles at
mutual distance two. Since the boundary of   a tile has length 24,
at least $60-23$ paths from the center cross the boundary in fans
for size at least two, and at  least $60-46$ paths between two
centers are red.
 Each red path has 13 edges between
the boundaries, so that there is an independent crossing if there
are at least 14 red paths between the centers of five tiles.

Graph    $G$ is a RAC graph. Therefore, draw each tile as a RAC
graph as shown in Fig.~\ref{fig:tile}, so that the edges between
grid points are parallel to the axis. Place the tiles from left to
right so that they are disjoint and their bottom row is on the
x-axis. There is a bundle of paths between any two vertices from
different tiles.
 Route each added path $P$ from a
vertex   to the outer face of the tile, so that internally there are
right angle crossings. This needs paths of length at most six, so
that the sixth edge is parallel to the $y$-axis. The middle edges
between the seventh and tenth vertex of each path are horizontal or
vertical and are routed outside the tiles, so that edges of two
paths do not overlap. If they cross, they cross at a right angle. If
there are several paths between two vertices, then they are
piecewise in parallel, except for the first and last edges.

It remains to show that there is  a graph  consisting of at least
five tiles at distance at least two. Consider a fan-crossing
embedding  $\mathcal{E}(G)$ of $k$ tiles, so that only the outer
boundary and the center in its interior  are taken into account. In
other words, there are $k$ Jordan curves and $k$ points, so that a
point is in the interior of a particular Jordan curve. The boundary
of a tile can be crossed at most 24 times by the boundaries of other
tiles, since it has length at most 24. In fact,  three or four edges
from both boundaries are involved in a crossing. Consider the planar
dual of $\mathcal{E}(G)$.
If center $c$ of tile $T$ is in face $f$, then assign $T$ to $f$.
Face $f$ has at most 25 adjacent faces, that are accounted to $f$,
since the boundary of a tile must be crossed for a new face and it
can be crossed   at most 24 times. In addition, there is the outer
face or the face from a  nesting tile enclosing $T$. Faces from
tiles in its interior are accounted to these tiles. Moreover, at
most five tiles can be assigned to a face, since a tile has diameter
four and any path from the center to a vertex in the boundary of $T$
can cross at most four  boundaries.
%
%
Suppose there are at least 650 tiles. If several tiles are assigned
to a face, then keep only one them. Each face has at most 25
neighbors with an assigned tile at distance one, so that there are
26 candidates. Hence, at least five tiles remain,   so that the
faces, to which they are assigned, have distance at least two.
%
\end{proof}

There should be simpler counterexamples, for example, a
  $10 \times 10$ tile with an added path  between the vertices at positions $(4,4)$ to $(7,7)$.
 We claim that this graph   is not
fan-crossing, since in any fan-crossing embedding, the end vertices
of the path
  are separated by a circle of crossed edges,
so that an internal edge of the path must cross such
an edge.\\

\begin{theorem} \label{thm:subdiv-appl}
If the crossing number exceeds the number $m$ of edges of graph $G$
so that   $cr(G) > 3 k m$, then the 3-subdivision $\sigma_3(G)$ is a
RAC graph, and thus fan-crossing free and quasi-planar, whereas
$\sigma_3(G)$ is not $k$-gap-planar and not $2k$-planar.
\end{theorem}

\begin{proof}
Graph $\sigma_3(G)$ is a RAC graph by
Theorem~\ref{thm:universal}(ii). It has at most  $3m$ edges, so that
there are at most $3km$ crossings in a $k$-gap-planar drawing of
$\sigma_3(G)$. Since the crossing number is preserved by
subdivisions, $\sigma_3(G)$ is not $k$-gap-planar if $cr(G) > 3km$.
Then it is not $2k$-planar by \cite{bbc-1gap-18}.
\end{proof}

Bae et al.~\cite{bbc-1gap-18}  have shown that the number of edges
in $n$-vertex $k$-planar graphs is at most $\max\{5.58\sqrt k,
17.17\} n$ and at most $5n-10$ for $k=1$. The Crossing Lemma and
$cr(K_{12})=150$ \cite{pr-crK11-07} provide simple examples of non
1-gap-planar graphs. Our example  is an alternative to the RAC and
non 1-planar graph by Eades and Liotta \cite{el-racg1p-13}, which
has 85 vertices and is constructed from $K_5$  by an addition of
many short paths similar to the graph in Theorem~\ref{thm:tiles}.

\begin{corollary}
The 3-subdivision of $K_{19}$ is  RAC and not 1-gap-planar
(2-planar). The 2-subdivision of $K_{12}$ is fan-crossing free and
not 1-gap-planar (2-planar)
\end{corollary}

Angelini et al.~\cite{abks-ldgbp-18} observed that cube-connected
cycles are not $k$-planar. Their arguments apply to $k$-gap-planar
  and   fan-crossing graphs, so that we obtain   alternative
  proofs for quasi-planar graphs that are neither $k$-gap-planar
  \cite{bbc-1gap-18}  nor fan-crossing
  \cite{abks-beyond-Kn-19, at-mneqpg-07, bddmpst-fan-15}.

\begin{theorem} \label{thm:CCC-not-FAN-GAP}
For every $k$, there are  cube-connected cycles  that are
fan-crossing free and  quasi-planar, respectively, and are not
 $k$-planar, $k$-gap planar, and fan-crossing, respectively
\end{theorem}

\begin{proof}
Cube-connected cycles are fan-crossing free and quasi-planar by
Theorem~\ref{thm:node-expansion-uni}. They are not $k$-(gap)-planar
if  the crossing number \cite{sv-cnCCC-93} exceeds $k$-times the
number of edges, which holds for $d$-dimensional cube-connected
cycles and $k \leq \frac{1}{30}\frac{2^d}{d}-4$. For fan-crossing
graphs use Lemma \ref{lem:fan-2planar}.
\end{proof}

\begin{corollary} \label{cor:counterexFAN}
$CCC^{11}$ is fan-crossing free and quasi-planar but neither
fan-crossing  nor 3-gap-planar nor 6-planar.
\end{corollary}
\begin{proof}
The lower bound for the crossing number of
$CCC^{11}$~\cite{sv-cnCCC-93} exceeds  the number of edges of
$CCC^{11}$
 by a factor $c>3.1$. Hence, it is not 3-gap-planar and thus not 6-planar
\cite{bbc-1gap-18} and not fan-crossing by Lemma
\ref{lem:fan-2planar}.
\end{proof}

Note that the cube $H^3$  is planar so that $CCC^3$ is planar.
 Figure~\ref{fig:CCC4} shows a 1-planar drawing of $H^4$,
 which can be transformed into a RAC drawing,  and an  IC-planar drawing $CCC^4$, so
that these graphs are 1-planar and RAC \cite{bdeklm-IC-16}. Every cube-connected cycle
is 1-bend RAC \cite{acddfs-pRAC-11}, that is $\sigma_1(CCC^d)$ is RAC. We can
summarize:

\begin{figure}[t]
  \centering
    \includegraphics[scale=0.6]{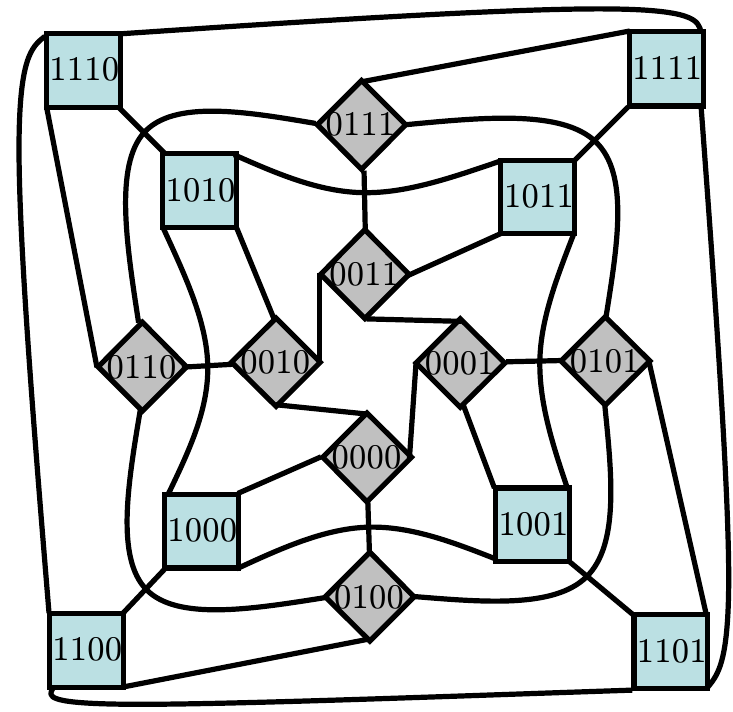}
  \caption{1-planar drawings of $H^4$ and $CCC^4$.
}
 \label{fig:CCC4}
\end{figure}

\begin{corollary} \label{cor:RAC-FAN}
Any two graph classes from   (a)-(c)  are incomparable: (a)
fan-crossing free and RAC, (b) fan-crossing, and (c) 1-gap-planar
and 2-planar.
\end{corollary}

Figure~\ref{fig:inclusion} displays the relationships   between
major beyond-planar graph classes. The shown inclusions were known,
as well as the incomparability between fan-crossing and 2-planar
resp. 1-gap planar graphs. We have added  incomparabilities for  the
fan-crossing-free and RAC graphs. Still open is the relationship
between the quasi-planar graphs and the fan-crossing free resp.
1-gap-planar graphs. We conjecture an incomparability.

\section{Properties of Fan-Crossing Free Graphs} \label{sect:properties}

It is well-known that the 4-clique $K_4$ admits two embeddings, as a
planar tetrahedron and 1-planar with a pair of crossing edges. The
5-clique $K_5$ has five embeddings \cite{hm-dcgmnc-92}, as shown in
Fig.~\ref{fig:allK5}. Only the so-called T-embedding
  in Fig.~\ref{fig:allK5}(a)  is fan-crossing free and even
1-planar. The embedding in Fig.~\ref{fig:allK5}(e) has an edge which
is crossed by the edges of  triangle. It is 1-gap-planar and
quasi-planar and not fan-crossing and not 2-planar. In fan-crossing
embeddings it can be transformed into
the Q-embedding of Fig.~\ref{fig:allK5}(c)~\cite{b-fan-20}.\\

\subsection{Unique Embeddings} \label{sect:uniqueembeddings}
Consider a fan-crossing free embedding of $K_6$, which is obtained
by placing the next vertex into one of the faces of the T-embedding
of $K_5$. There are two possibilities up to symmetry: in a face with
or without a crossing point. Only the latter results in a
fan-crossing free embedding. The obtained embedding is unique, since
the edges must be routed as shown in Fig.~\ref{fig:K6}. Otherwise,
an edge is crossed by at least two edges of a fan.  The embedding
can be drawn with two or three vertices in the outer face. The
7-clique is not fan-crossing free, since it has too many edges, but
fan-crossing (fan-planar) \cite{bddmpst-fan-15}, see
Fig.~\ref{fig:K7-fan-2planar}. We summarize:

\begin{lemma} \label{lem:uniqueK5K6}
The cliques $K_5$ and $K_6$ are fan-crossing free   and each has a
unique fan-crossing free embedding.
\end{lemma}

\begin{figure}[t]
   \centering
   \subfigure[ ]{
     \includegraphics[scale=0.4]{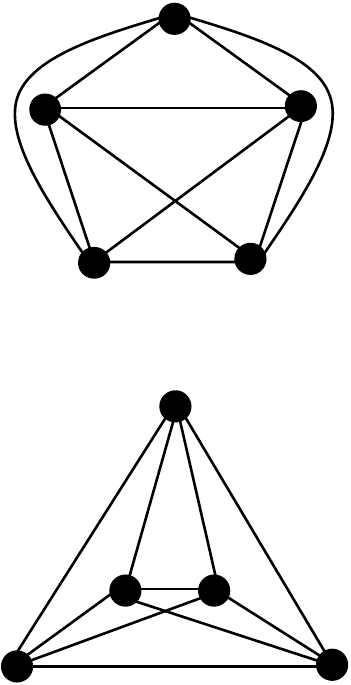}
   }
   \quad
      \subfigure[ ]{
     \includegraphics[scale=0.4]{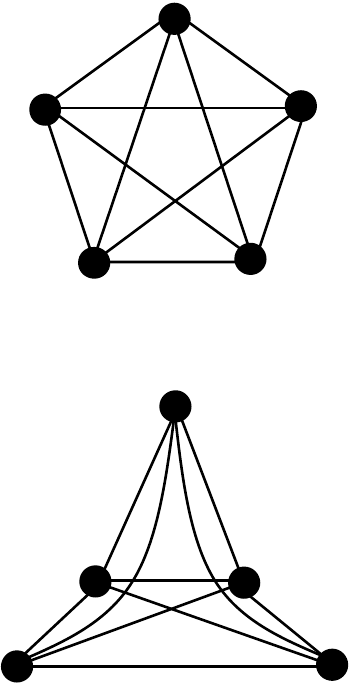}
   }
   \quad
      \subfigure[ ]{
     \includegraphics[scale=0.4]{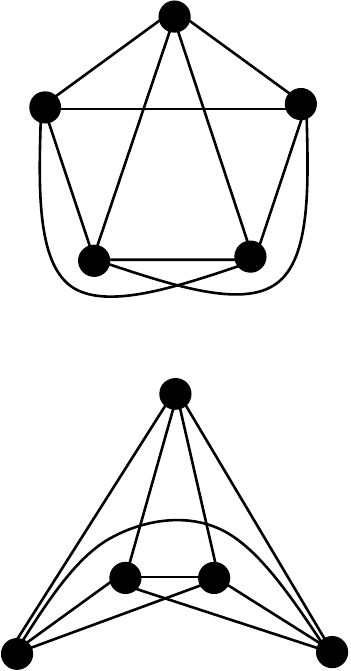}
   }
   \quad
      \subfigure[ ]{
     \includegraphics[scale=0.4]{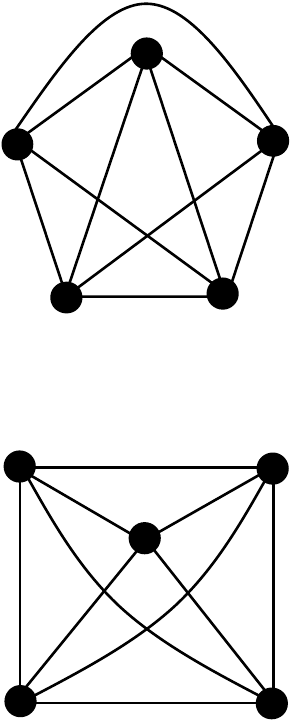}
   }
\quad
   \subfigure[ ]{
     \includegraphics[scale=0.4]{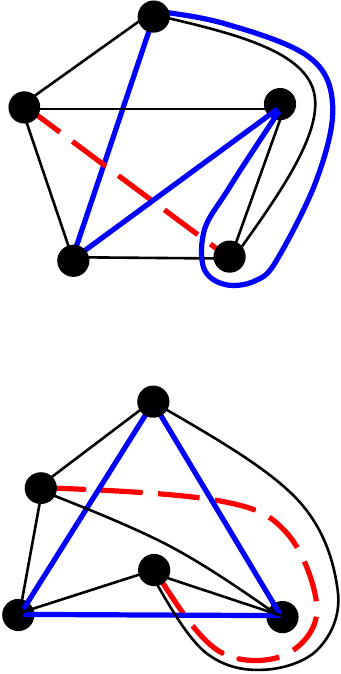}
   }
   \caption{All non-isomorphic embeddings of $K_5$ \cite{hm-dcgmnc-92}, each with two drawings.
   Only the T-embedding (a) is
    1-planar and fan-crossing free.
}
   \label{fig:allK5}
\end{figure}

\begin{figure}[t]
  \centering
   \subfigure[] {
     \includegraphics[scale=0.3]{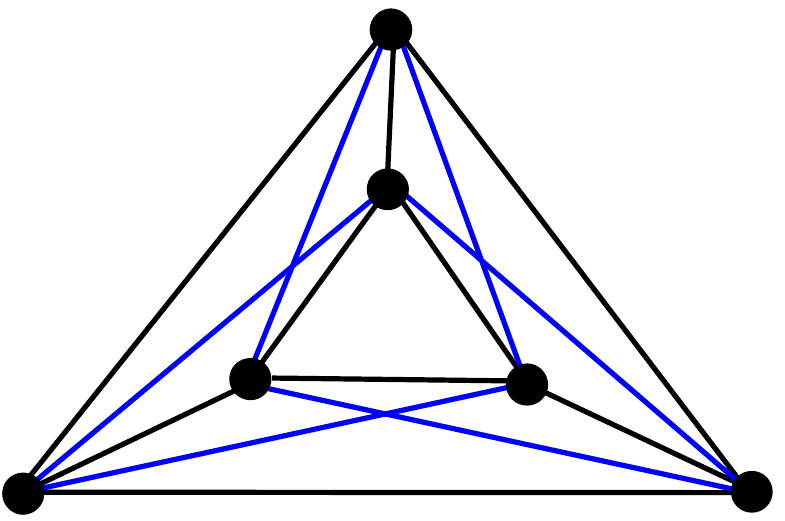}
    \label{fig:K6-triangle}
  }
  \hspace{5mm}
   \subfigure[] {
     \includegraphics[scale=0.3]{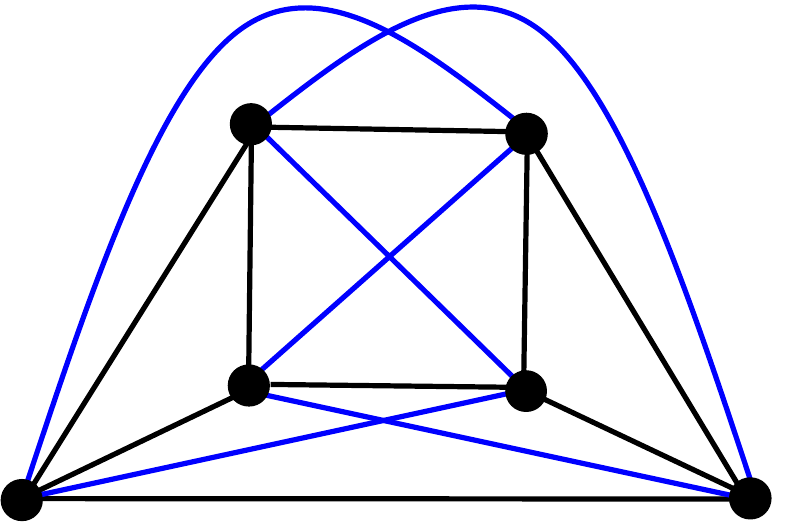}
    \label{fig:K6-Wconf}
  }
  \hspace{15mm}
\subfigure[] {
     \includegraphics[scale=0.4]{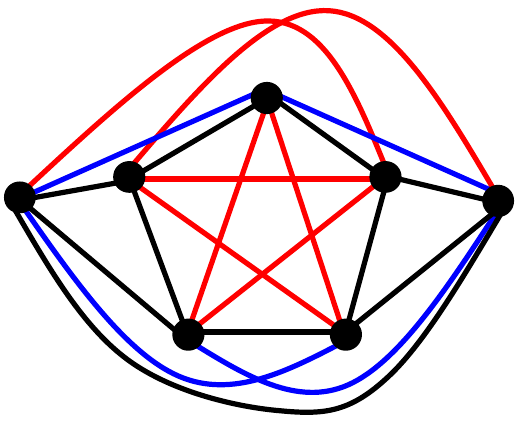}
    \label{fig:K7-fan-2planar}
  }
  \caption{
  The   fan-crossing free embedding of $K_6$ drawn (a) with a triangle and
  (b) with only two vertices in the outer face. (c) A 2-planar, fan-crossing embedding of
  $K_7$. Edges are colored black, blue, and red if they are
  uncrossed and crossed once and twice, respectively.
  }
  \label{fig:K6}
\end{figure}

Cheong et al.~\cite{cpkk-fan-15} have shown that every fan-crossing
free embedding of an $n$-vertex graph with $4n-8$ edges is 1-planar.
Such graphs are called extreme or \emph{optimal 1-planar}
\cite{bsw-1og-84}. The embedding consists of a 3-connected planar
subgraph,  in which each face is a quadrangle with a pair of
crossing edges. Such graphs exist for $n=8$ and for all $n \geq 10$
\cite{bsw-1og-84}. They have a unique 1-planar embedding, except for
  extended wheel graphs $XW_{2k}$ for $k \geq 3$, which have two
embeddings, where the poles are exchanged \cite{s-rm1pg-10}.  An
\emph{extended wheel graph} $XW_{2k}$ consists of two poles $p$ and $q$ and
a circle of length $2k$. There is an edge between each pole and each
vertex of the circle, whereas there is no edge $\edge{p}{q}$. In
addition, there is an edge between a vertex of the circle and the
vertex after next (in cyclic order). Each of these edges is crossed
by an edge incident to a pole. Note that optimal 1-planar graphs can
be recognized in linear time \cite{b-ro1plt-18}.

\begin{corollary} \label{cor:unique-optimal}
Every $n$-vertex fan-crossing free graph with $4n-8$ edges has a
unique embedding, except for the extended wheel graphs $XW_{2k}$,
which have two embeddings.
\end{corollary}

The \emph{crossed nested triangle graph} $\Delta_k$ consists of $k$
nested triangles $T_1,\ldots, T_k$ with vertices $a_i, b_i, c_i$ for
$i=1,\ldots,k$, so that each quadrangle between two successive sides
of the triangles has a pair of crossing edges and is $K_4$, as shown
in Fig.~\ref{fig:nested-hermit}. Triangle $T_i$ is at \emph{level}
$i$, where $T_1$ is in the outer face.  The subgraph induced by two
consecutive triangles is $K_6$, which is an \emph{inner} $K_6$ if it
is  induced by $T_i, T_{i+1}$ for $i=2,\ldots, k-2$. Graph
$\Delta_k$ has $3k$ vertices and $12k-9$ edges.  It is 1-planar and
admits a straight-line drawing, but it is not a RAC graph, since it
has $K_6$ subgraphs.

\begin{figure}[t]
  \centering
     \includegraphics[scale=0.5]{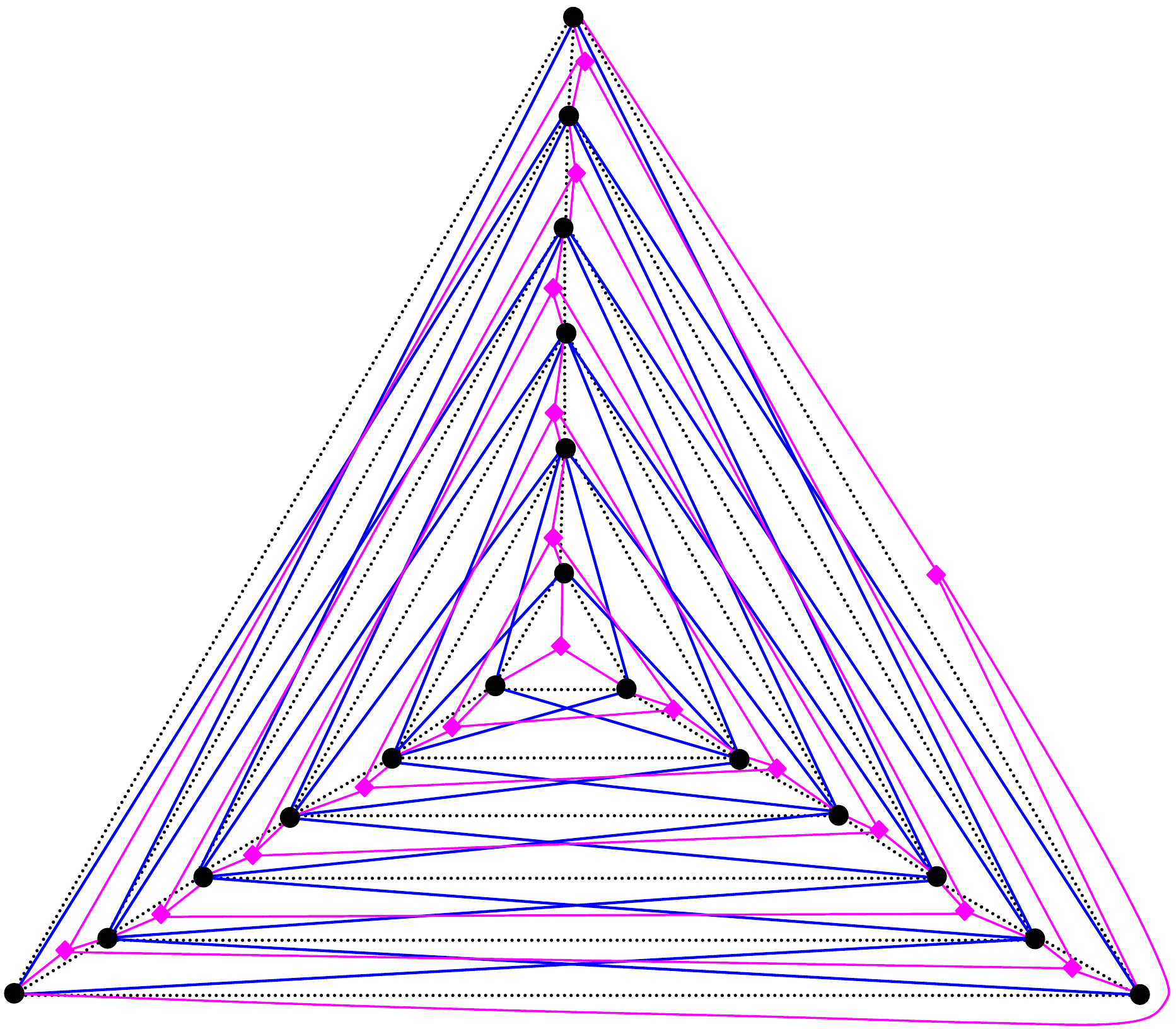}
  \caption{
  A fan-crossing free embedding of a nested triangle graph,
  augmented by  ``hermits'', drawn by  pink squares, for the proof of
  Theorem~\ref{thm:sparse-fcf}. The drawing is not quasi-planar. A re-routing of the pink
  edges makes it quasi-planar.
  }
  \label{fig:nested-hermit}
\end{figure}

\begin{lemma} \label{lem:uniqueXDelta-graph-fcf}
For every $k \geq 1$, the crossed nested triangle graph $\Delta_k$
has a unique fan-crossing free embedding.
\end{lemma}

\begin{proof}
For $k=1$ there is a triangle, and $K_6$ for $k=2$ has a unique
fan-crossing free embedding by Lemma~\ref{lem:uniqueK5K6}. For $k
\geq 3$, the restriction of $\Delta_k$ to two consecutive triangles
is $K_6$. Each such $K_6$ has a unique fan-crossing free embedding,
so that it is drawn with a triangle in its outer face. Otherwise,
there is a   W-configuration, as shown in Fig.~\ref{fig:K6-Wconf},
which enforces a crossing of an edge by at least two edges of a fan
with the common vertex in $K_6$ and the other vertices on the
previous or next level.
 Hence, all triangles are drawn as nested triangles. Clearly,
 one can choose the outer face of the drawing
 with three vertices or with two vertices and a crossing point.
\end{proof}

Note that the crossed nested triangle graph $\Delta_k$ admits many
fan-crossing drawings, since a $K_6$ has many fan-crossing
embeddings with different embeddings of its $K_5$ subgraphs.\\

Next we consider a traversal of $K_5$ or $K_6$ by a path in a
fan-crossing free embedding, which causes a delay  by  at least
three bends. This is used in the NP-hardness proof of Theorem
\ref{thm:NP-fcf}.

\begin{lemma} \label{lem:K5fantraverse}
In every fan-crossing free embedding, if   path $P$ traverses
 $K_5$ $(K_6)$, so that $P$ and the clique have disjoint sets of
 vertices
  and there are different edges of  $K_5$ $(K_6)$
 crossed first and last by $P$,  then at least four (five) edges of
$P$ cross edges of $K_5$ $(K_6)$.
\end{lemma}

\begin{proof}
Let $P = (u_0, \ldots, u_t)$ be a path such that $u_0$ and $u_t$ are
in the outer face and $u_1, \ldots, u_{t-1}$ are in an inner
triangle of the $T$-embedding of   $K_5$. Let the first edge $(u_0,
u_1)$ of $P$ cross edge $(v_1, v_2)$ of $K_5$ and let the last edge
$(u_{t-1}, u_t)$ of $P$ cross $(v_1, v_3)$ in a fan-crossing free
embedding of a graph $G$ that includes  $P$ and   $K_5$. Since $K_5$
has a unique fan-crossing free embedding,  $P$ must cross two more
edges of $K_5$, since a simultaneous crossing of two or more edges
induces a fan-crossing. Hence, $P$ needs at least four edges for a
traversal of $K_5$. The case of $K_6$ is similar.
%
\end{proof}

\subsection{Extremal Graphs} \label{sect:extremal}

A graph $G$ is \emph{maximal} for a class of graphs if no edge can
be added without violation, so that $G+e$ is not in the class. The
\emph{density} is the maximum number of edges of all maximal
$n$-vertex graphs in a class. The \emph{sparsity} is the minimum.
Density and sparsity coincide for planar graphs, whereas they differ
for some classes of beyond-planar graphs. Brandenburg et
al.~\cite{begghr-odm1p-13} have shown that the sparsity of 1-planar
graphs is at most $\frac{45}{17}n - \frac{84}{17}$.  Such graphs are
obtained by so called \emph{hermits}, which are vertices of degree
two that cannot be linked to other vertices without violation.
Similar results have been obtained for 2-planar
\cite{begghr-odm1p-13}, IC, NIC  and outer 1-planar graphs
\cite{abbghnr-o1p-16, bbhnr-NIC-17}.

\begin{theorem} \label{thm:sparse-fcf}
There are maximal fan-crossing free graphs with $m = 7/2n - 17/2$
edges for every $n=6k+1$ with $k \geq 3$.
\end{theorem}
\begin{proof}
Consider a crossed nested triangle graph $\Delta_k$ for $k \geq 3$.
 Attach three hermits $h_i, h'_i$ and $h''_i$ along the edges
$\edge{a_i}{a_{i+1}}, \edge{b_i}{b_{i+1}}$ and $\edge{c_i}{c_{i+1}}$
of each triangle  $T_i = (a_i, b_i, c_i)$ for $i=1,\ldots,k-1$ and
connect each hermit with the  vertices of the edge to which it is
attached. Also link the hermits by edges $\edge{h_i}{h'_i},
\edge{h'_i}{h''_i}$ and $\edge{h_i}{h''_i}$. Add a hermit in the
outer face and link it to $a_1, b_1, c_1$, and similarly for the
inner face, see Fig.~\ref{fig:nested-hermit}.

Graph $H_k$ is fan-crossing free, as shown by
Fig.~\ref{fig:nested-hermit}, where the embedding is not
quasi-planar. We claim that it is   maximal fan-crossing free.
Therefore, observe that $\Delta_k$ is maximal fan-crossing free and
1-planar and has a unique fan-crossing free embedding.
Consider   hermit $h_i$ that is attached to edge
$\edge{a_i}{a_{i+1}}$.  Then $h_i$ can be placed into the faces $f$
and $f'$ to either side of $\edge{a_i}{a_{i+1}}$. It cannot be
placed into another face without creating a fan-crossing by two
edges. Also $h_i$ cannot be linked to another vertex of $\Delta_k$
without creating a crossing by a fan of at least two edges. The case
of $h'_i$ and $h''_i$ is similar. Hermit $h_i$ is linked to $h'_i$
and $h''_i$ by an edge, but not to any other hermit if the embedding
is fan-crossing free. Then at least two edges incident to a vertex
of $\Delta_k$ must be crossed. If three hermits $h_i, h'_i, h''_i$
are simultaneously placed into another face that is next to the one
in which they were placed, then the edges between them create a
crossing by a fan of two edges, as illustrated in
Fig.~\ref{fig:hermitF1}. Similarly, the hermit in the inner face can
be placed in a neighboring face, as illustrated in
Fig.~\ref{fig:hermitF2}, but it cannot be linked to another hermit
or another vertex without violation. Hence, the fan-crossing free
embedding of $H_k$ is unique, as claimed.

Graph $\Delta_k$ has $3k$ vertices and $12k-9$ edges. We have added
$3(k-1)+2$ hermits and $9(k-1)+6$ edges incident to the hermits.
Hence,  graph $H_k$ has  $6k-1$ vertices and $21k-12$ edges, so that
$m = 7/2n - 17/2$.
\end{proof}

\begin{figure}[t]
  \centering
   \subfigure[] {
    \includegraphics[scale=0.35]{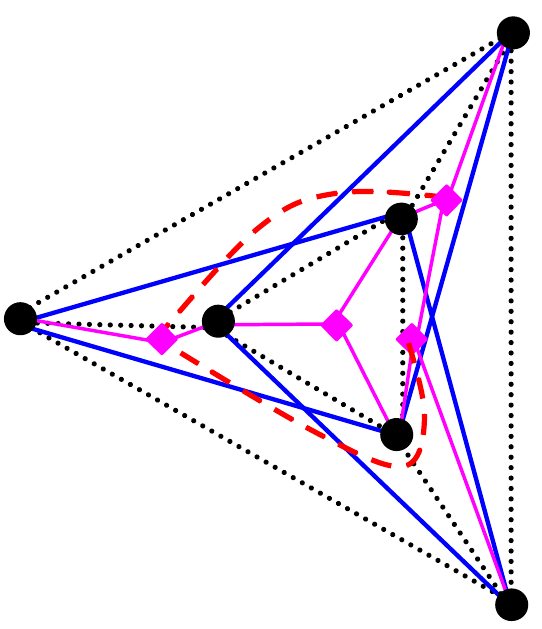}
    \label{fig:hermitF1}
  }
  \hspace{10mm}
   \subfigure[] {
    \includegraphics[scale=0.35]{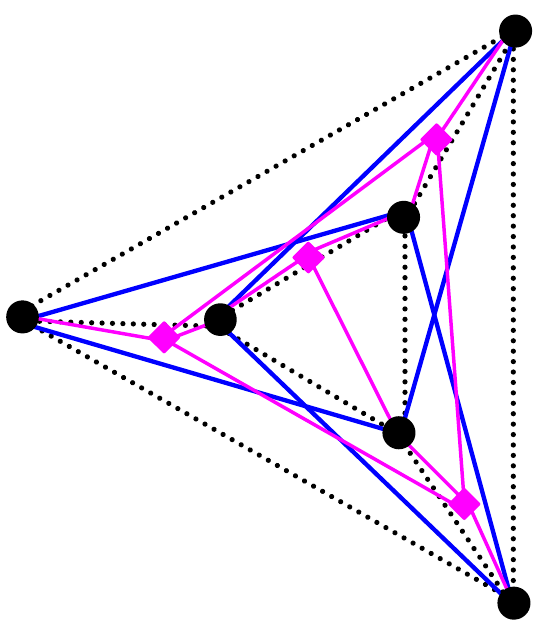}
    \label{fig:hermitF2}
  }
  \caption{Illustration to the proof of
  Theorem~\ref{thm:sparse-fcf} if hermits are moved to other
  faces.
  }
  \label{fig:nested-fcf-F}
\end{figure}

\subsection{The Complexity of Fan-Crossing Free Recognition}
\label{sect:NP}

The recognition of   beyond-planar graphs is  mostly NP-hard, see
\cite{dlm-survey-beyond-19}.
%
%
 It has been expected that   recognizing fan-crossing free graphs
is NP-complete, too, which has been stated as an open problem in
\cite{dlm-survey-beyond-19} and is proved next.

\begin{theorem} \label{thm:NP-fcf}
The recognition of fan-crossing free graphs is NP-complete.
\end{theorem}
\begin{proof}

Clearly, the problem is in NP.

For the NP-hardness we adapt   the construction given by Grigoriev
and Bodlaender \cite{GB-AGEFCE-07} for $1$-planar graphs. The proof
is by reduction from 3-PARTITION, which is a strongly NP-hard
problem \cite{gj-cigtn-79}: An instance $I$ of $3$-PARTITION
consists of a multiset $A$ of $3m$ positive integers with $B/4 < a <
B/2$ and $\sum_{a \in A} =  mB$ for some integer $B$ and   $a \in
A$. The $3$-PARTITION problem asks whether $A$ can be partitioned
into $m$ subsets $A_1, \ldots, A_m$, each of size three, such that
the sum of the numbers in each subset  equals $B$.

From $I$ we construct a graph $G_I$ in polynomial time, so that $I$
has a solution if and only if $G_I$ admits a fan-crossing free
embedding. The components of the reduction are  a transmitter, a
collector, $m$ splitter and ``fat'' edges. A fat edge $\edge{u}{v}$
is a $K_5$ with vertices $u$ and $v$ and three more vertices, where
$u$ and $v$ are  in the   outer face if there is a fan-crossing free
drawing. For an illustration see Fig.~\ref{fig:NPreduction}.

The transmitter is a double-wheel with a center  $c_T$    and two
circuits of length $3m$. Let $C = (u_1, \ldots, u_{3m})$ be the
outer and $C' = (u'_1, \ldots, u'_{3m})$ the inner circuit. Let
$\edge{c_T} {u'_})$ and $\edge{u'_i}{ u_i}$ be fat edges  for each
$i$. Then the center $c_T$ is a vertex of each of the $K_5$ graphs
of the fat edges that are incident to the center. Each sector with
boundary $u_i, u'_i, c_T, u'_{i+1}, u_{i+1}$ is partitioned into
three triangles by (normal) edges $\edge{u'_i}{ u'_{i+1}},
\edge{u_i} {u_{i+1}}$ and a diagonal $\edge{u_i}{u'_{i+1}}$. The
triangle $(u_i, u'_i, u'_{i+1})$ is called a \emph{sector-triangle},
and $(u_i, u_{i+1}, u'_{i+1})$ is called an \emph{outer triangle},
see the enlargement in Fig.~\ref{fig:NPreduction}(b). The collector
is defined accordingly, with a center $c_C$ and $Bm$ sectors with an
outer boundary $(v_1, \ldots, v_{Bm})$. The boundary of the
transmitter is partitioned into $m$ segments, each of width three,
and accordingly, the collector has $m$ segments, each of width $B$.
The ends of the segments are connected by a fat edge in circular
order. Hence, there is fat edge $(u_{3i}, v_{Bi})$ for $i=1,\ldots,
m$.

For each element of $A$ there is a splitter $P_a$  with center $c_a$
and $a+1$ satellites at distance two from the center, i.e., each
satellite has a connector on the path to $c_a$. In
\cite{GB-AGEFCE-07}, the satellites are at distance one from the
center. We connect one satellite of each splitter with the
transmitter center $c_T$ and the remaining satellites with the
collector center $c_C$.
Hence, there is a path of length six   from $c_T$ via each $c_a$ to
$c_C$.

Since $3$-PARTITION is strongly NP-hard, all numbers $a \in A$ can
be given in unary encoding, such that   $G_I$ has polynomial size
and can be constructed in polynomial time.

For the  correctness of the reduction we follow  the arguments given
in \cite{GB-AGEFCE-07}. However, we use  $K_5$ instead of $K_6$ as a
fat  edge, and sector triangles  and satellites at distance two from
the center of each splitter.

If the instance of $3$-PARTITION has a solution, then the
transmitter, collector and regions are drawn as illustrated in
Fig.~\ref{fig:NPreduction}. The splitters of a $3$-set   with
$a_1+a_2+a_3=B$ appear in one region and between the outer
boundaries of the transmitter and the collector. Since each sector
of the collector has $B$ sectors, we can put a single satellite in
each sector triangle and obtain a 1-planar drawing.

Conversely,   suppose there is a fan-crossing free drawing of $G_I$,
where the drawing is on the sphere. As in  \cite{GB-AGEFCE-07}, we
place $c_T$ at the North Pole and $c_C$ at the South Pole.
A \emph{meridian path} $M_i$ is a path of  five fat edges between
the centers of the transmitter and the collector through the
vertices $u_{3i}$ and $v_{Bi}$ on the boundaries, for $i=1,\ldots,
m$. Each meridian path also contains a unique path of length five
from $c_T$ to $c_C$. A \emph{splitter path} $SP_a = (c_T, a_T,
a'_T,c_a, a'_C, a_C, c_C)$ is a path of length six between the
centers of the transmitter and the collector via  the center of a
splitter $P_a$, such that   $a_T$ and $a_C$ are satellites of $P_a$
and $a'_T$ and $a'_C$ are vertices between $a_T$ and $c_a$ and $c_a$
and $c_C$, respectively. Note that there is a splitter path from
$c_C$ through $a_C$ and $c_a$ for $a$ vertices and there is a single
path from $c_T$ through $a_T$ to $c_a$.

From now on assume that we are given a fan-crossing free embedding
of $G_I$.
 First, observe, that
two meridian paths $P$ and $P'$ do not cross, i,e., there is no
crossing of two  edges $e$ and $e'$, where $e$ is in the subgraph
induced by $P$ and $e'$ is in the subgraph induced by $P'$.
 Edges $e$ and $e'$ belong to two distinct $K_5$, and a
fan-crossing is unavoidable by the uniqueness of a fan-crossing free
embedding of $K_5$ from Lemma~\ref{lem:uniqueK5K6} if $P$ and $P'$
cross.

Second, a meridian path $P$ and a splitter path $SP_a$ do not cross.
Towards a contradiction, suppose that edge $e$ of $SP_a$ crosses an
edge $f$ of some fat edge of $P$. Each splitter path must cross the
boundaries of the transmitter and  the collector, which needs a path
of length at least four. If, in addition, an edge of a $K_5$ is
crossed, then at least four more edge  is needed. However, splitter
paths are tight and have length six.

In consequence, we can follow the arguments given by Grigoriev and
Bodlaender. Two successive meridian paths $M_i$ and $M_{i+1}$ define
a face, so that there is a cyclic ordering of faces according to the
circuits of the boundaries of the transmitter and the collector. In
consequence, there is a unique way to draw the transmitter around
the North Pole, and similarly, there is a unique way to draw the
collector around the South Pole.

We now consider the drawing of splitters. Each splitter path has
length six. It takes a fan-crossing free path of length at least
three from the center of the transmitter to cross the the inner and
outer circuit including the diagonal, and similarly for the
collector. Hence, the center $c_a$ of a splitter must be placed
between the outer circuits from $c_T$ and $c_C$.

 In consequence, each satellite
is placed in a sector triangle and each connector in  an outer
triangle. Then there is at most one satellite in each sector
triangle; otherwise there is a fan crossing, since the edges to the
satellites are incident to the centers $c_T$ and $c_C$ of the
transmitter and collector, respectively.

Hence, for each region, we have exactly three splitters, which each
have one satellite in a sector triangle, and there is at most one
satellite in each of the $B$ sector triangles of the region (face)
of the collector. Since we have    $Bm$  such paths, each face must
contain exactly three splitters with exactly $B$ paths between the
splitter centers and the South Pole, which implies that the instance
$I$ has a solution.
%
\end{proof}

\begin{figure}[t]
   \centering
   \subfigure[the general schema ]{
      \rotatebox{270}{%
     \includegraphics[scale=0.45]{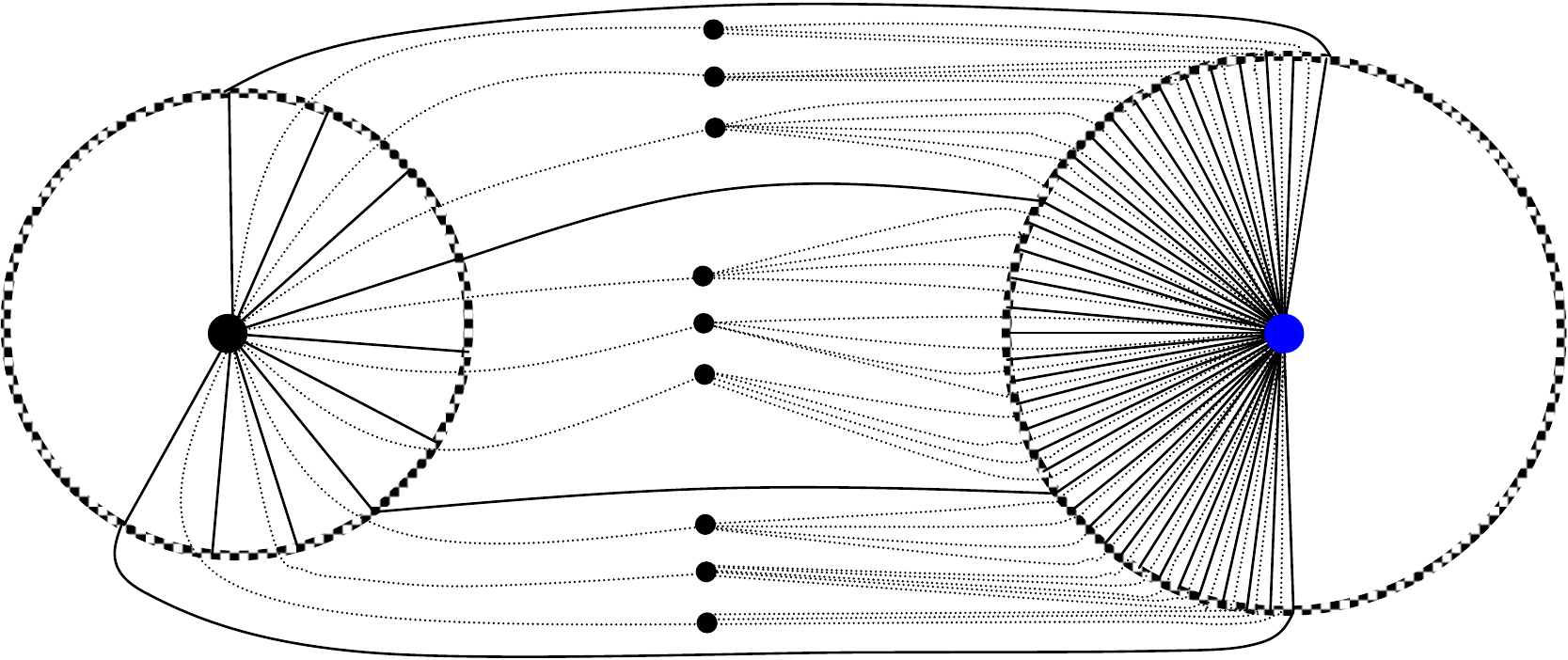}
   }
   }
   \quad
      \subfigure[a detailed view]{
     \includegraphics[scale=0.55]{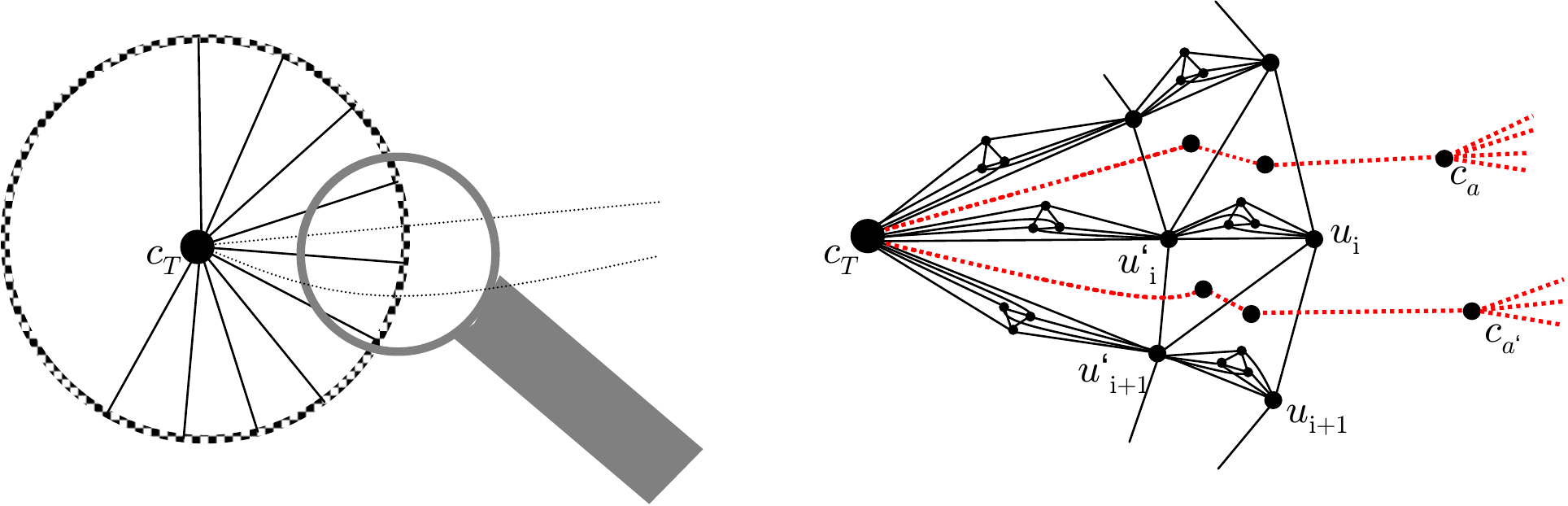}
   }

   \caption{Illustration of graph $G_I$ for the NP-reduction}
   \label{fig:NPreduction}
\end{figure}

The recognition problem for fan-crossing free graphs with a fixed
rotation system is also NP-complete. Here, a graph is given with a
rotation system describing the cyclic ordering of the edges incident
to each vertex. Therefore, we can modify the proof of
\cite{abgr-1prs-15} for the NP-completeness of   recognizing
1-planar graphs with a rotation system in a similar way as in the
proof of Theorem~\ref{thm:NP-fcf}.

\section{Conclusion}

In this work, we have shown that the  fan-crossing free graphs are
incomparable with the fan-crossing, 2-planar, and 1-gap-planar
graphs, respectively. There were not listed in the table of
relationships  in
\cite{dlm-survey-beyond-19}.  It remains to show that there are
fan-crossing free graphs that are not quasi-planar graphs. Moreover,
we have added fan-crossing free graphs to the list of beyond-planar
graph classes with an NP-complete recognition problem,  which  was
open so far \cite{dlm-survey-beyond-19}.

We have shown a closure resp. non-closure of graph classes under
   subdivision, node-to-circle expansion, and  path-addition, respectively, and
   have introduced the notion of universality, which is used for a
   specification of beyond-planarity.

 Many new graph classes can be defined by the intersection of two
 (or more) graph classes or by combining the respective properties
 in a drawing. Such "double" classes are largely unexplored, as remarked in \cite{dlm-survey-beyond-19}.
 Some facts are known about 1-planar and RAC graphs, where IC-planar
 combines both properties \cite{bdeklm-IC-16}. However, there are
 graphs that are 1-planar and RAC and are not IC-planar, for example,  $n$-vertex
tiles, which have too many edges for IC-planarity. There are
NIC-planar that are 1-planar and not RAC \cite{bbhnr-NIC-17}. On the
other hand, $n$-vertex RAC graphs with $4n-10$ edges
\cite{el-racg1p-13} and $n$-vertex fan-crossing free graphs with
$4n-8$ edges \cite{cpkk-fan-15} are 1-planar. There are non 1-planar graphs that are
both fan-crossing and fan-crossing free \cite{b-fan+fcf-18}. As last, the 2-planar
fan-crossing  graphs are characterized as 5-map graph if the
graphs are clique augmented \cite{b-5maps-19}.

 How shall we draw fan-crossing free graphs? What are their
 stack  and queue numbers? How are fan-crossing free graphs related to classes of
 graphs that are defined by visibility representations \cite{DEGLST-bkvg-07, ekllmw-b1vg-14}?
 There are partial answers to these questions for 1-planar graphs \cite{abk-sld3c-13,
 abk-bt1pg-15, bbkr-bt1planar-17, b-vr1pg-14, b-Tshape-18}.

\section{Acknowledgements}

I wish to thank Therese Biedl for point out the simple
box-visibility representation, and to Michael Bekos for drawing my
attention to the crossing number of cube-connected cycles.


\bibliographystyle{abbrv}
\bibliography{brandybibV8a}

\end{document}